\documentclass[11pt]{article}
\usepackage{fullpage}
\usepackage{amssymb,amsmath,amsfonts}
\usepackage{graphicx}
\usepackage{amsmath}
\usepackage{ifpdf}
\ifpdf\fi
\usepackage{subfigure}
\usepackage{amsmath,amsthm, amssymb,color,graphicx,epsf,verbatim}

\newcommand{\para}[1]{\medskip \noindent {\bf #1}}

\newlength{\figwidth}
\newlength{\twowidth}
\usepackage{algorithm}
\usepackage{algorithmicx}
\usepackage[noend]{algpseudocode}
\usepackage{times}

\newtheorem{example}{Example}

\newcommand{\eat}[1]{}
\begin{document}

\author{
 Edith Cohen,\, Graham Cormode,\, Nick Duffield \\
{AT\&T Labs--Research}\\
{180 Park Avenue}\\
{Florham Park, NJ 07932, USA}\\
{\{edith,graham,duffield\}@research.att.com}}

  \eat{
\alignauthor Edith Cohen\\
       \affaddr{AT\&T Labs--Research}\\
       \affaddr{180 Park Avenue}\\
       \affaddr{Florham Park, NJ 07932, USA}\\
       \email{edith@research.att.com}
\alignauthor  Graham Cormode\\
       \affaddr{AT\&T Labs--Research}\\
       \affaddr{180 Park Avenue}\\
       \affaddr{Florham Park, NJ 07932, USA}\\
       \email{graham@research.att.com}
\alignauthor Nick Duffield\\
       \affaddr{AT\&T Labs--Research}\\
       \affaddr{180 Park Avenue}\\
       \affaddr{Florham Park, NJ 07932, USA}\\
       \email{duffield@research.att.com}
   }

\def\varopt{\textnormal{\sc VarOpt}}
\def\poly{\textnormal{poly}}
\def\WBS{\textnormal{\sc WB}}
\def\HT{\textnormal{\sc HT}}
\newcommand\E{\textsf{E}}
\newcommand{\var}{\mathop{\sf Var}}
\def\Lvaropt{\textnormal{\sc Lvaropt}} 

\newtheorem{theorem}{Theorem}
\newtheorem{lemma}[theorem]{Lemma} 
\newtheorem{definition}[theorem]{Definition} 
\newcommand{\VS}{{V\Sigma}}
\newcommand{\SV}{{\Sigma}V}
\newcommand{\LCA}{\textnormal{\sc LCA}}
\def\qd{Q-digest}
\def\cR{\mathcal{R}}
\def\cL{\mathcal{L}}
\def\cK{\mathcal{K}}
\newcommand{\ignore}[1]{}
\newcommand{\onlyinproc}[1]{#1}
\newcommand{\notinproc}[1]{}

\title{Structure-Aware Sampling: \\ Flexible and Accurate Summarization}
 \maketitle

\begin{abstract}
In processing large quantities of data, a fundamental problem is to
obtain a summary which supports approximate query answering.  Random
sampling yields flexible summaries which naturally support subset-sum
queries with unbiased estimators and well-understood confidence
bounds.

Classic sample-based summaries, however, are designed for arbitrary
subset queries and are oblivious to the structure in the set of keys.
The particular structure, such as hierarchy, order, or product space
(multi-dimensional), makes {\em range queries} much more relevant for
most analysis of the data.

Dedicated summarization algorithms for range-sum queries have also been
extensively studied.  They can outperform existing sampling schemes in
terms of accuracy on range queries per summary size. Their accuracy,
however, rapidly degrades when, as is often the case, the query spans
multiple ranges.  
They are also less flexible---being targeted for range sum
queries alone---and are often quite costly to build and use.

In this paper we propose and evaluate variance optimal sampling
schemes that are {\em structure-aware}.  
These summaries improve over 
the accuracy of existing {\em structure-oblivious} sampling schemes on
range queries while retaining the benefits of sample-based summaries:
flexible summaries, with high accuracy on 
both range queries and arbitrary subset queries.
\end{abstract}

\section{Introduction}
Consider a scenario where a large volume of data is collected on a
daily basis: for example, sales records in a retailer, or network
activity in a telecoms company. 
This activity will be archived in a warehouse or other storage
mechanism, but the size of the data is too large for data analysts 
to keep in memory. 
Rather than go out to the full archive for every query, it is natural
to retain accurate summaries of each data table, and use these queries
for data exploration and analysis, reducing the need to read through
the full history for each query. 
Since there can be many tables (say, one for every day at each store,
in the retailer case, or one for every hour and every router in the
network case), we want to keep a very compact summary of each table,
but still guarantee accurate answers to any query. 
The summary allows approximate processing of queries, 
in place of the original data (which may be slow to access or even no
longer available); it also allows fast `previews' of computations
which are slow or resource hungry to perform exactly.

\begin{example}
\label{eg:motivate}
As a motivating example, consider network data in the form
of IP flow records.
Each record has a source and destination IP address, a port number, and
size (number of bytes).  
IP addresses form a natural hierarchy, where prefixes or sets of prefixes 
define the ranges of interest.  
Port numbers indicate the generating application, and 
related applications use ranges of port numbers. 
Flow summaries are used for many network management tasks,
including planning routing strategies, and traffic anomaly detection.
Typical ad hoc analysis tasks may involve 
 estimating the amount of traffic between different
subnetworks, or the fraction of VoIP traffic on a certain network. 
Resources for collection, transport, storage and analysis of network
measurements are expensive; therefore, 
structure-aware summaries are needed by network operators to
understand the behavior of their network. 
\end{example}

Such scenarios have motivated a wealth of work on data summarization and
approximation.
There are two main themes: 
methods based on random sampling, and algorithms that build more complex
summaries (often deterministic, but also randomized). 
Both have their pros and cons.
Sampling is fast and efficient, and has
useful guaranteed properties.
Dedicated summaries can offer
greater accuracy for the kind of range queries which are most common
over large data, albeit at a greater cost to compute, and providing
less flexibility for other query types. 
Our goal in this work is to provide summaries which combine the best
of both worlds: fast, flexible summaries which are very accurate for
the all-important range queries. 
To attain this goal, we must understand existing methods in detail to
see how to improve on their properties. 

Summaries which are based on random sampling 
allow us to build (unbiased) estimates of properties of the data set, 
such as counts of individual identifiers (``keys''), sums of weights
for particular subsets of keys, and so on, all specified after the
data has been seen. 
Having high-quality estimates of these primitives allows us to
implement higher-level applications over samples, such as
computing order statistics over subsets of the data, heavy hitters
detection, longitudinal studies of trends and correlations, and so
on. 

Summarization of items with weights  traditionally uses Poisson
sampling, where each item is sampled independently. 
The approach which sets the probability of including an item in the
sample to be proportional to its weight (IPPS)~\cite{Hajekbook1981}
enables us to use the Horvitz-Thompson estimator~\cite{HT52}, 
which minimizes the sum of per-item variances.
``\varopt'' samples~\cite{Cha82,Tille:book,varopt:CDKLT08}
improve on Poisson samples in that the sample size is fixed and 
they are more accurate on subset-sum queries. 
In particular \varopt\ samples have 
{\em variance optimality}: they achieve variance over the
queries that is provably the smallest possible for any sample of that
size. 

Since sampling is simple to implement and flexible to use, it is
 the default summarization method for large data sets. 
Samples support a rich class of possible queries directly, such as
those mentioned in Example~\ref{eg:motivate}:
evaluating the query over the sampled data (with appropriately
scaled weights) usually provides an unbiased, low variance estimate of
the true answer, while not requiring any new code to be written. 
These summaries provide not only estimates of aggregate
values but also a representative sample of keys that satisfy a
selection criteria.  
The fact that estimates are unbiased also means that relative
error decreases for queries that span multiple
samples or larger subsets and the estimation error is governed by
exponential tail bounds:  the estimation error, in terms of the number
of samples from any particular subset, is highly concentrated around
the square root of the expectation.

We observe, however, that traditionally sampling
has neglected the inherent structure that is present,
 and which is known before the data is observed. 
That is, data typically exists within a well-understood schema 
that exhibits considerable structure. 
Common structures include {\em order} where there is a natural ordering
over keys; {\em hierarchy} where keys are leaves within a hierarchy
(e.g. geographic hierarchy, network hierarchy); and
combinations of these where keys are multi-dimensional points in a 
{\em product structure}.
Over such data, queries are often {\em structure-respecting}.  
For example, on ordered data with $n$ possible key-values, 
although there are $O(2^n)$ possible subset-sum queries, 
the most relevant queries may be the $O(n^2)$ possible range queries. 
In a hierarchy, relevant queries may correspond to particular nodes 
in the hierarchy
(geographic areas, IP address prefixes), which represent
$O(n\log n)$ possible ranges. In a product structure, likely queries are 
boxes---intersections of ranges of each dimension. 
This is observed in Example~\ref{eg:motivate}: the queries mentioned
are based on the network hierarchy. 

While samples have been shown to work very well for queries which
resemble the sums of {\em arbitrary} subsets of keys, 
they tend to be less satisfactory when restricted to range queries. 
Given the same summary size, samples can be out-performed in
accuracy by dedicated methods such as 
(multi-dimensional) histograms~\cite{GKTD:sigmod2000,PoosalaIoannidis:VLDB97,LKC:sigmod1999},
wavelet
transforms~\cite{Matias98wavelet-basedhistograms,Vitter98datacube},
and geometric
summaries~\cite{SuriTZ:CG2004,HSST:ISAAC04,BCEG:ToA2007,HHH:CKMS:TKDD2008,ZSSDC:IMC04}
including the popular 
Q-digest~\cite{ShrivastavaBAS:sensys04}.

These dedicated summaries, however, have inherent drawbacks: they
primarily support queries that are sum aggregates over the original
weights, 
and so other queries must be expressed in terms of this primitive. 
Their accuracy rapidly degrades when the query spans multiple
ranges---a limitation since natural queries 
may span several (time, geographic) ranges within the same summary
and across multiple summaries.
Dedicated summaries do not provide ``representative'' keys of selected
subsets, and require changes to existing code to utilize.  
Of most concern is that they can be very slow to compute, requiring a
lot of I/O (especially as the dimensionality of the data grows): 
a method which gives a highly accurate summary of each hour's data
is of little use if it takes a day to build!
Lastly, the quality of the summary may rely on certain structure being
present in the data, which is not always the case. 
While these summaries have shown their value in efficiently
summarizing one-dimensional data (essentially, arrays of counts),
their behavior on even two-dimensional data is less satisfying:
troubling since this is where accurate summaries are most needed. 
For example, in the network data example, we are often interested in
the traffic volume between (collections of) 
various source and destination ranges.

\ignore{
In the context of the IP flow example, 
deterministic summaries taken with respect
to source IP addresses work well for capturing the volume under all
sufficiently large prefixes.  
But if we are interested in a union of prefixes 
(say, associated with a certain geographic location or type of
customer) 
or in gleaning from daily summaries the total monthly volume of a
certain prefix, 
errors add up and the relative accuracy deteriorates. 
In contrast, with an unbiased sample-based summaries estimates 
the relative error can decrease.  
For similar reasons, in estimating the total amount of traffic associated
with low-volume addresses, deterministic summaries perform poorly with
no useful guarantee, while the (unbiased) sample promises a good
estimate if the total volume is large enough. 
}

Motivated by the limitations of dedicated summaries, 
and the potential
for improvement over existing (structure-oblivious) sampling schemes, 
we aim to design sampling schemes that are both \varopt\ and
{\em structure-aware}.
At the same time, we aim to match the accuracy of deterministic summaries on
range sum queries and retain the desirable 
properties of existing sample-based summaries: unbiasedness, tail
bounds on arbitrary subset-sums, flexibility and support for
representative samples, and good I/O performance. 

\subsection{Our Contributions}
\noindent
We introduce a novel algorithmic sampling framework, 
which we refer to as {\em probabilistic aggregation}, 
for deriving \varopt\ samples.   
This framework makes explicit the freedom of choice in building a
\varopt\ summary which has previously been overlooked. 
Working within this framework, we design {\em structure-aware} 
 \varopt\ sampling schemes which exploit this freedom 
 to be much more accurate on ranges than 
 their structure-oblivious counterparts.

\begin{list}{\labelitemi}{\leftmargin=1em}
\item
For hierarchies, we design an efficient algorithm that constructs
\varopt\ summaries with bounded ``range discrepancy''.
That is, for any range, the number of samples deviates 
from the expectation by less than $1$.  
This scheme has the minimum possible variance on ranges of any
unbiased sample-based summary. 
\item
For ordered sets, where the ranges consist of all intervals, 
we provide a sampling algorithm which builds a \varopt\ summary 
with range discrepancy less than $2$. 
We prove that this is the best possible for any \varopt\ sample.
\item 
For $d$-dimensional datasets, we propose sampling algorithms where
the discrepancy between  $p(R)$, the expected number of 
sample points in the range $R$, and the actual number is 
$O(\min\{d s^{\frac{d-1}{2d}},$ $
\sqrt{p(R)}\})$, 
where $s$ is the sample size.
\end{list}
This improves over structure-oblivious random sampling, 
where the corresponding discrepancy is $O(\sqrt{p(R)})$.

Discrepany corresponds to error of range-sum queries but
sampling has an advantage over other summaries
with similar error bounds:
The error on queries $Q$ which span multiple ranges grows linearly with the
number of ranges for other summaries but has square root dependence for samples. Moreover, for samples the expected error never exceeds $O(\sqrt{p(Q)})$ (in expectation) regardless of the number of ranges.

\smallskip
\noindent
{\bf Construction Cost.}
For a summary structure to be effective, it must be possible to
construct quickly, and with small space requirements. 
Our main-memory sampling algorithms perform tasks such as sorting
 keys or (for multidimensional data) building a kd-tree.  
We propose even cheaper alternatives which perform two read-only
 passes over the dataset using memory that depends on the desired
 summary size $s$ (and independent of the size of the data set).  When
 the available memory is $O(s\log s)$, 
we obtain a \varopt\ sample that with high probability 
$1-O(1/\poly s)$ is close in quality to the algorithms which store and
manipulate the full data set. 

\smallskip
\noindent
{\bf Empirical study.}  
To demonstrate the value of our new
  structure-aware sampling algorithms, 
we perform experiments comparing to popular 
  summaries, in particular the wavelet
  transform~\cite{Vitter98datacube}, 
  $\epsilon$-approximations~\cite{HSST:ISAAC04}, 
  randomized sketches~\cite{ccf:icalp2002} 
and to structure-oblivious random sampling.  
These experiments show that it is possible to have the best of both
worlds: summaries with equal or better accuracy than the
best-in-class, which are flexible and dramatically more efficient to construct and
work with.

\section{Probabilistic aggregation}   \label{probagg:subsec}
\onlyinproc{
This section introduces the ``probabilistic aggregation'' technique. 
For more background, see the review  of core concepts from
sampling and summarization in Appendix~\ref{prelim:sec}.

Our data is modeled as a set of (key, weight) pairs:
each key $i\in K$ has weight $w_i\geq 0$. 
A sample is a random subset $S\subset K$.  
A sampling scheme is IPPS when, for expected sample size $s$ and
derived threshold $\tau_s$, the sample includes key $i$
with probability $\min\{w_i/\tau_s,1\}$.  
IPPS can be acheived with  
  {\em Poisson} sampling (by including keys independently) 
or \varopt\ sampling, which allows correlations between key inclusions
to achieve improved variance and fixed sample size of exactly $s$. 
There is not a unique \varopt\ sampling scheme, but rather
there is a large family of \varopt\ sampling distributions:
 the well-known ``reservoir sampling'' is a
special case of \varopt\ on a data stream with uniform weights.   
Classic tail bounds, including Chernoff bounds, apply both to
\varopt\ and Poisson sampling.  

 Structure is specified as a {\em range space} $(\cK,\cR)$ with $\cK$ being the key domain and {\em ranges} $\cR$ that are subsets of $\cK$.
 The {\em discrepancy} $\Delta(S,R)$ of a sample $S$ on a range $R\in \cR$ 
is the difference between the number of sampled keys $S\cap R$ and its expectation $p(R)$.   We use $\Delta$ to denote the maximum discrepancy over all ranges $\cR$.
Disrepancy $\Delta$ means that the error of range-sum queries is at most $\Delta \tau_s$.  If a sample is Poisson or \varopt, it follows from Chernoff bounds that the expected discrepancy is $O(\sqrt{p(R)})$ and (from bounded VC dimension of our range spaces) that the maximum range discrepancy is $O(\sqrt{s\log s})$ with probability $O(1-\poly(1/s))$.  With structure-aware sampling, we aim for much lower discrepancy.
}

\para{Defining probabilistic aggregation.}
Let  $p$ be the vector of sampling probabilities.
We can view a sampling scheme that picks a set of keys $S$ as
operating on $p$.   
Vector $p$ is incrementally modified: setting $p_i$ to 1 means
$i$ is included in the sample, while $p_i=0$ means it is
omitted. 
When all entries are set to 0 or 1, the sample is chosen (e.g. Poisson
sampling independently sets each entry to 1 with probability $p_i$). 
To ensure a \varopt\ sample,
the current vector $p'$ must be a {\em
  probabilistic aggregate} of the original $p$.

\ignore{
The first step in our approach to making sampling structure aware is
to define a relaxation of \varopt\ sampling which 
allows for {\em partial} inclusions of keys.
These ``partial'' 
\varopt\ samples, which we refer to as {\em probabilistic aggregates}, 
can always be completed to generate a \varopt\ sample of the original
data.
The input and output are thought of as vectors of (sampling) probabilities.
}
 
A random vector $p^{(1)}\in [0,1]^n$ 
is a {\em probabilistic aggregate} of a vector $p^{(0)}\in [0,1]^n$ if
the following conditions are satisfied:
\begin{itemize}
\item [(i)] ({\it Agreement in Expectation})
$\forall i, \E[p^{(1)}_i]=p^{(0)}_i$, 
\item [(ii)] ({\it Agreement in Sum})
$\sum_i p^{(1)}_i = \sum_i p^{(0)}_i$, and
\item [(iii)] ({\it Inclusion-Exclusion Bounds})
\begin{align*}
\mbox{(I):} && \E[\prod_{i\in J} p^{(1)}_i]  & \quad \leq \quad \prod_{i\in J} p^{(0)}_i \\
\mbox{(E):} && \E[\prod_{i\in J} (1-p^{(1)}_i)] & \quad \leq \quad \prod_{i\in J} (1-p^{(0)}_i )\ .
\end{align*}
\end{itemize}

\begin{algorithm}[t]
\caption{{\sc Pair-Aggregate}$(p,i,j)$ }\label{pairagg:alg}
\begin{algorithmic}[1]
\Require $0  < p_i,p_j < 1$
\If {$p_i+p_j< 1$}
\State {\bf if} {$rand()< \frac{p_i}{p_i+p_j}$}
  {\bf then}
 $p_i \gets p_i+p_j$;  $p_j \gets 0$
\State{\bf else } $p_j \gets p_i+p_j$;  $p_i \gets 0$
\Else \Comment{$p_i+p_j\geq 1$}
\State{\bf if} {$rand() < \frac{1-p_j}{2-p_i-p_j}$}
{\bf then} $p_i \gets 1$;  $p_j \gets p_i+p_j-1$ 
\State{\bf else} $p_i \gets p_i+p_j-1$; $p_j \gets 1$
 \Comment{w/prob $\frac{1-p_i}{2-p_i-p_j}$}
\EndIf
\State \textbf{return} $p$
\end{algorithmic}
\end{algorithm}

We obtain \varopt\ samples by performing a sequence of
probabilistic aggregations, each setting
at least one of the probabilities to 1 or 0. 
In Appendix~\ref{probaggapp} we show that probablistic aggregations
are transitive and that set entries remain set.  Thus, such a process
must terminate with a \varopt\ sample. 


\smallskip
\noindent
{\bf Pair Aggregation.}
Our summarization algorithms perform
a sequence of simple aggregation steps which we refer to as
{\em pair aggregations} (Algorithm~\ref{pairagg:alg}).
Each pair aggregation step modifies only two entries 
and sets at least one of them to $\{0,1\}$.  
The input to {\em pair aggregation} is a
vector $p$ and a pair $i,j$ with each $p_i,p_j\in (0,1)$.  The output
vector agrees with $p$ on all entries except $i,j$ and one of the
entries $i,j$ is set to $0$ or $1$.  It is not hard to verify, separately considering cases $p_i+p_j<1$ and $p_i+p_j\geq 1$,
that {\sc Pair-Aggregate}$(p,i,j)$ correctly computes 
a probabilistic aggregate of its input, and hence the sample is \varopt.

Pair aggregation is a powerful primitive.
It produces a sample of size exactly $s=\sum_i p^{(0)}_i$ 
\footnote{Assuming that 
 $\sum_i p^{(0)}_i$ is integral.  
This can be ensured (deterministically) by choosing $\tau$ as
described in Algorithm \ref{get_tau_k:alg}.}.
Observe that the choice of which pair $i,j$ to aggregate at any point
can be arbitrary---and the result is still a \varopt\ sample.  
This observation is what enables our approach.   
We harness this freedom in pair selection
to obtain \varopt\ samples that are structure
aware:  
Intuitively, by choosing to aggregate pairs that are
 ``close'' to each other with respect to the structure, 
we control the range impact of the ``movement'' of probability mass.

\section{One dimensional structures} \label{onedim:sec}

We use pair aggregation to make sampling structure-aware by describing
ways to pick which pair of items to aggregate at each step. 
For now, we assume  the data fits in main-memory, 
and our input is the list of keys and their associated IPPS
probabilities $p_i$.
We later
discuss the case when the data exceeds the available memory.

For hierarchy structures
 (keys $\cK$ are associated with leaves of a tree and $\cR$ contains
all sets of keys under some internal node) we show how to obtain
\varopt\ samples with (optimal) maximum range discrepancy $\Delta <1$.
There are two special cases of hierarchies:
   (i) {\em disjoint ranges} (where
  $\cR$ is a partition of $\cK$)---captured by a flat 2-level 
hierarchy with parent nodes corresponding to ranges and
(ii)~{\em order} where there is a linear order on keys and 
$\cR$ is the set of all prefixes---the corresponding 
hierarchy is a path with single leaf below each internal node.
For order structures  where $\cR$ is the
set of ``intervals'' (all consecutive sets of keys) we show that there
is always a \varopt\ sample with maximum range discrepancy $\Delta<2$ and
prove that this is the best possible.

\notinproc{
\subsection{Disjoint ranges and hierarchy} 

\noindent
We specify pair selection criteria which result in a range-optimal summary:
}

\label{drandhierarchy:sec}
\begin{list}{\labelitemi}{\leftmargin=1em}
\item
{\bf Disjoint ranges:}
Pair selection picks pairs where both keys belong to the same range $R$.  
When there are multiple choices, we may choose one arbitrarily.
Only when there are none do we
select a pair 
that spans two different ranges (arbitrarily if there are multiple
choices).

\item 
{\bf Hierarchy:}
   Pair selection picks pairs with lowest $\LCA$ (lowest common ancestor).  
That is, we pair aggregate $(i,j)$ if there are no other pairs
 with an $\LCA$ that is a descendant of $\LCA(i,j)$.  
\end{list}

\begin{figure}[t]
\centering
\includegraphics[width=0.9\columnwidth]{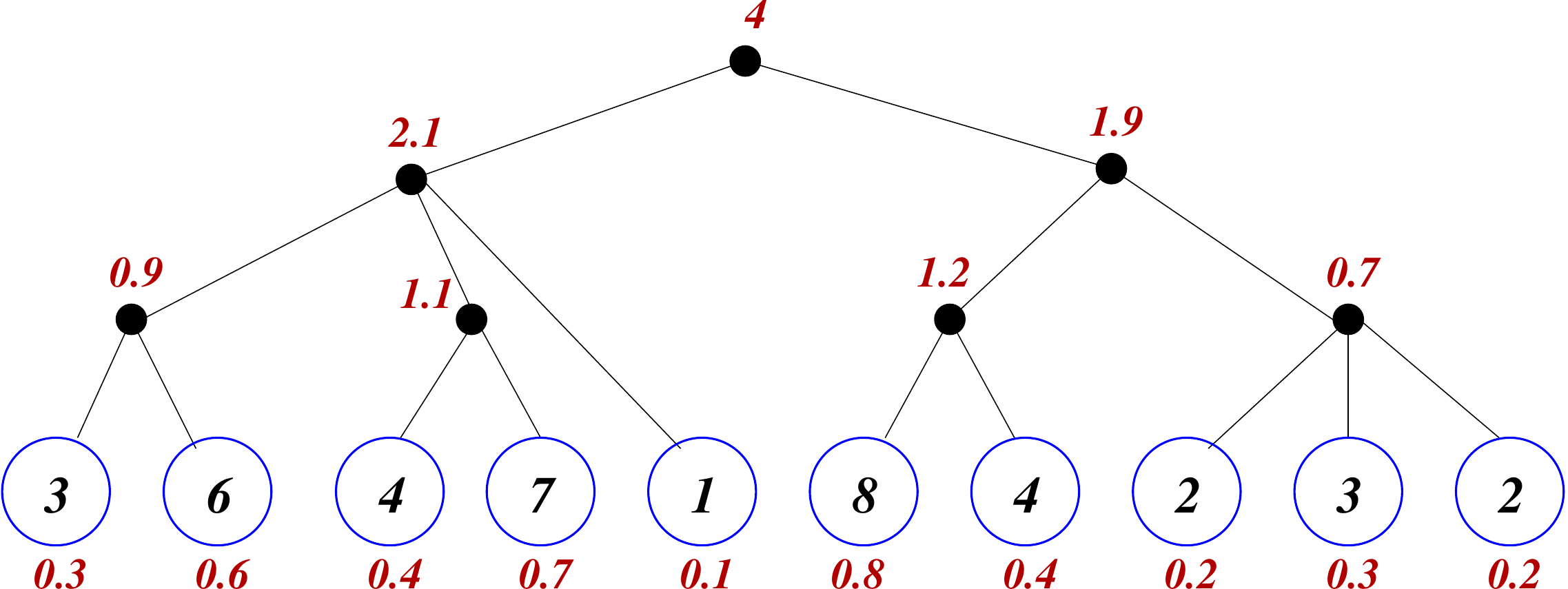}
\vspace{0.2in}

\centering
\begin{tabular}{c|llllllllll}
leaf & 1 & 2 & 3 & 4 & 5 & 6 & 7 & 8 & 9 & 10 \\
\hline
IPPS & 0.3 & 0.6 & 0.4 & 0.7 & 0.1 & 0.8 & 0.4 & 0.2 & 0.3 & 0.2 \\
\hline
(1)+(2),(3)+(4), & 0 & 0.9 & 1 & 0.1 & 0.1 & 0.2 & 1 & 0.5 & 0 & 0.2 \\
(6)+(7),(8)+(9) & \multicolumn{10}{c}{} \\
(2)+(4), (8)+(10) &  0 & 1 & 1 & 0 & 0.1 & 0.2 & 1 & 0 & 0 & 0.7 \\
(6)+(10) & 0 & 1 & 1 & 0 & 0.1 & 0 & 1 & 0 & 0 & 0.9 \\ 
(5)+(10) & 0 & 1 & 1 & 0 & 0 & 0 & 1 & 0 & 0 & 1 
\end{tabular}
\caption{Sampling over a hierarchy structure}
\label{probaggH:fig}
\end{figure}

Following these rules guarantees low range discrepancy:
they ensure that
for all ranges $R\in\cR$ and for all iterations $h$ where $R$ has at 
least one entry which is not set, we have 
$\sum_{i\in R} p^{(h)}_i \equiv \sum_{i\in R} p^{(0)}_i$.
So,
 at termination, 
when all entries in $R$ are set:
\[\textstyle |S\cap R| \in \{\lfloor \sum_{i\in R} p^{(0)}_i \rfloor, \lceil
\sum_{i\in R} p^{(0)}_i \rceil \}.\] 
Hence, the maximum range discrepancy is $\Delta < 1$. 
\onlyinproc{
In Appendix~\ref{mrangeH:sec} we bound the discrepancy of multi-range queries.}

\begin{example}
Figure~\ref{probaggH:fig} demonstrates sampling when the key domain is a hierarchy.  The tree structure shown is induced by 
keys in the data set (those with positive weights).  Each key corresponds to
a leaf node which shows its weight and IPPS sampling probability for sample size $s=4$. Each internal tree node shows the expected number of samples under it.

All \varopt\ samples include exactly $4$ keys but if not
structure aware then the number of samples under internal nodes 
can significantly deviate from their expectation.
A Poisson IPPS sample has 4 keys in expectation, and is 
oblivious to structure.

The table shows a sequence of pairwise aggregations which follows
the hierarchy pair selection rule. The result is a structure-aware
\varopt\ sample, consisting of the leaves $S=\{2,3,7,10\}$. 
One can verify that the number of samples under each internal nodes is 
indeed the floor or
ceiling of its expectation.
\end{example}

\onlyinproc{
\begin{list}{\labelitemi}{\leftmargin=1em}
\item {\bf Order structures:}
  In Appendix~\ref{ordered:sec} we establish the following:
\end{list}
\begin{theorem} \label{orderoffline:thm}
For the order structure (all intervals of ordered keys), 
(i) there always exists a \varopt\ sample
distribution with maximum range discrepancy $\Delta\leq 2$.  
(ii) For any
fixed $\Delta<2$, there exist inputs for which a \varopt\ sample
distribution with maximum range discrepancy $\leq \Delta$ does not
exist. 
\end{theorem}


}

\section{Product structures} \label{product:sec}

We now consider summarizing $d$-dimensional data, where the key structure
 projected on each axis is an order or a hierarchy. 
Ranges are axis parallel hyper rectangles: a product of one-dimensional 
key ranges (order) and/or internal nodes of a hierarchy.

We develop a \varopt\ sampling algorithm where 
the discrepancy on a range $R$ is
that of a (structure oblivious) \varopt\ sample on a subset with 
$\mu\leq \min\{p(R),2ds^{\frac{d-1}{d}}\}$.  
Hence, the estimation error is subject to tail bounds \eqref{chernoff:upper} and \eqref{chernoff:lower} with this value of $\mu$ and 
concentrated around 
$\sqrt{\mu}\leq \min\{\sqrt{p(R)},\sqrt{2d}s^{\frac{d-1}{2d}}\}$.   
\notinproc{
\subsection{Algorithm}}

As in the one-dimensional case, 
the intuition behind our approach is to limit the range discrepancy 
by preferring pairwise aggregations that result in ``localized'' 
movement of ``probability mass.''     

\medskip
\noindent
{\bf Uniform case.}
We start with the special case of a 
uniform distribution over a $d$-dimensional hyper cube with measure $s=h^d$.
Our algorithm partitions the hypercube into $s$ unit cells and 
selects the sample by independently picking 
a single point uniformly from each unit cell. 
The resulting sample $S$ is a \varopt\ sample (of size $s$) of the
uniform hypercube. 
 For analysis, observe that any  axis-parallel hyperplane intersects at most 
$h^{d-1}=s^{\frac{d-1}{d}}$ unit cells.
Therefore, 
any axis-parallel box query $R$ intersects at most 
$2d s^{(d-1)/d}$ cells that are not contained in $R$. 
The only error in our estimation comes from these 
``boundary'' cells which we denote $B(R)$:
all other cells are either fully inside or fully outside $R$, and so
do not contribute to the discrepancy. 
We map each boundary cell $C\in B$ to a 0/1 random variable which is $1$
with probability proportional to the size of the overlap, $|C\cap R|$.  
These random variable are independent Poisson 
with $\mu=\sum_{C\in B} |C\cap R| \leq \min\{p(R),|B(R)|\}$, and so
the tail bounds hold.

\medskip
\noindent
{\bf General Case.}
In general, the probability mass is not distributed uniformly
throughout the space as in the previous case. 
So instead, we seek to build a partition of the space into regions so
that the probability mass is (approximately) equal. 
In particular, we consider using 
kd-trees to obtain a partition into cells containing keys whose
sum of probabilities is  $\Theta(1)$ 
(in general it is not possible to partition the discrete probabilities
to sum to exactly 1). 
Choosing kd-trees means that 
every axis-parallel hyperplane intersects $O(s^{\frac{d-1}{d}})$ cells.
Since cells are not exact units, 
we have to carefully account for aggregations of the
 ``leftover'' probabilities.

Let $K$ be the input set of weighted $d$-dimensional keys. 
Then:

\begin{list}{\labelitemi}{\leftmargin=1em}
\item
 Compute IPPS inclusion probabilities for $K$ and
set aside all keys with $p_i=1$  
(they must all be included in the summary).  
Hence, wlog, we have that all keys in $K$ have $p_i< 1$.
\item 
Compute a hierarchy $T$ over the (multidimensional) keys in $K$:
$T\leftarrow \mbox{{\sc KDhierarchy}}(0,K)\ .$
\item 
Apply the hierarchy summarization algorithm (Section~\ref{onedim:sec}) to $T$.
\end{list}

Algorithm {\sc KD-hierarchy} builds a kd-tree, splitting on each
dimension in turn. 
 At each internal node we select a hyperplane perpendicular
 to the current axis that partitions the probability weight in half 
(or as equally as possible.) 
 Each leaf of the tree then corresponds to an undivided rectangle
containing approximately unit probability mass. 
\onlyinproc{Analysis and examples are given in Appendix~\ref{KDanal:sec}.}


\begin{algorithm}[t]
\caption{{\sc KD-hierarchy}$(depth,key\_set)$}\label{kdhierarchy:alg}
\begin{algorithmic}[1]
\If {$|\mbox{key\_set}|=1$}
\State  $\mbox{\em h.val} \leftarrow$ key\_set
\State  $\mbox{\em h.left} \leftarrow$ null; $\mbox{\em h.right} \leftarrow$ null; 
\State  \textbf{return $h$} \Comment{kd-hierarchy $h$ is a leaf node}
\Else
\State $\mbox{\em h.val} \leftarrow$ null
\State a $\leftarrow$ depth mod d  \Comment{axis on which partition is made}
\If {axis $a$ has an order structure}
\State $m \leftarrow \arg\min_m \left\vert \sum_{i | key_a(i)\leq
  m}p_i-\sum_{i | key_a(i)> m}p_i\right\vert$ 
\Statex \Comment{$m$ is weighted median of key\_set ordered on  axis $a$}
\State left\_set $\leftarrow \{i | key_a(i)\leq m\}$;
\State right\_set $\leftarrow \{i | key_a(i)> m\}$;
\Else \Comment{axis $a$ has a hierarchy structure $H_a$}
\State Find the partition of key\_set into left\_set and right\_set
over all linearizations of the hierarchy to minimize
$\left\vert \sum_{i\in \mbox{left\_set}}p_i-\sum_{i\in \mbox{right\_set}} p_i\right\vert$
\EndIf
\State {\em h.left}$\leftarrow \mbox{{\sc KD-hierarchy}}(depth+1,\mbox{left\_set})$
\State {\em h.right}$\leftarrow  \mbox{{\sc KD-hierarchy}}(depth+1,\mbox{right\_set})$
\State \textbf{return} $h$ 
\EndIf
\end{algorithmic}
\end{algorithm}

\section{I/O efficient sampling}
\label{IO:sec}

\begin{algorithm}[t]
\caption{{\sc IO-aggregate}$(i)$}\label{io:alg}
{\bf Process key $i$:}
\begin{algorithmic}[1]
\State 
$p_i \leftarrow \min\{1, w_i/\tau_s\}$ \Comment{IPPS sampling prob}
\If {$p_i=1$} 
 \State $S\leftarrow S\cup \{i\}$ \Comment{$i$ is placed in the sample}
\Else \Comment{$p_i<1$}
 \State $L\leftarrow L(i)$ \Comment{$L$ is the cell that contains $i$}
 \If {$a(L)=\emptyset$} \Comment{Cell $L$ has no key with $p_a\in \{0,1\}$)}
 \State $a(L) \leftarrow i$ \Comment{$i$ becomes the active key of its cell}
 \Else \Comment{$L$ has an existing active key}
\State  $a \leftarrow a(L)$ 
\State {\sc Pair-Aggregate}$(p,i,a)$\Comment{One of $p_i,p_a$ is set}
  \State $a(L)\leftarrow \emptyset$
 \If {$p_a = 1$} 
  \State $S\leftarrow S\cup\{a\}$\Comment{$a$ is placed in the sample}
 \EndIf
 \If {$0<p_a<1$}
  \State $a(L) \leftarrow a$\Comment{$a$ remains the active key of $L$}
 \EndIf
 \If {$p_i = 1$} 
  \State $S\leftarrow S\cup\{i\}$\Comment{$i$ is placed in the sample}
 \EndIf
 \If {$0<p_i<1$}
  \State $a(L) \leftarrow i$\Comment{$i$ becomes the active key of $L$}
 \EndIf
 \EndIf
 \EndIf
\end{algorithmic}
\end{algorithm}

The algorithms presented in previous sections
assume that we can hold the full data set in memory to generate the
summary. 
As data sets grow,
we require summarization methods that are more I/O efficient. 
In particular, the reliance on being able to sort data, locate data in
hierarchies, and build kd-trees over the whole data may not be realistic for large data
sets. 
In this section, 
we present I/O efficient alternatives that generate a structure-aware
\varopt\ sample
while only slightly compromising on range discrepancy with respect to the
main-memory variants. 
The intuition behind our approach is that a structure-oblivious
\varopt\ sample of sufficient size $\tilde{O}(s)$ 
is useful to guide the construction of a structure-aware summary because with high probability it hits
{\em all} sufficiently large ranges (those with $p(R)\geq 1$)
(In geometric terms, it forms an 
$\epsilon$-net of the range space~\cite{epsilonnets:1987}).  
Once built, the summary can be kept in fast-access storage while the
original data is archived or deleted.

\medskip
\noindent
{\bf Description.}
Our algorithms perform two read-only streaming 
passes over the (unsorted) input dataset.
When using memory of size $\tilde{O}(s)$ (where $s$ is the desired
sample size), the range discrepancy is similar to that of the main
memory algorithm with high probability.
In the first pass we compute a random sample $S'$ of size $s' >
s$ using memory $s'$ via Poisson IPPS or stream
\varopt\ sampling (reservoir sampling if keys have uniform weights). 
We also compute the IPPS threshold value $\tau_s$ for a sample of size $s$
(described in Appendix~\ref{prelim:sec}).
After completion of the first pass (in main memory) we use $S'$
to construct a partition $\cL$ of the key domain.  
The partition has the property that with high probability
$p(L)\leq 1$ for all $L\in\cL$.

In the second pass, we incrementally 
build the sample $S$, initialized to $\emptyset$. 
We 
perform probabilistic aggregations, guided by the partition $\cL$,
using IPPS probabilities for a sample of size $s$. 
We maintain at most one {\em active key} $a(L)$ for each cell $L\in\cL$, which is initialized to null.
Each key $i$ is processed using {\sc IO-aggregate} (Algorithm \ref{io:alg}): if $p=\min\{1,w_i/\tau_s\}=1$, then $i$ is added to $S$. Otherwise,
if there is no active key in the cell $L(i)$ of
$i$, then $i$ becomes the active key.  
If there is an active key $a$, we {\sc Pair-Aggregate} $i$ and $a$. 
If the updated $p$ value of one of them becomes $1$, we include this
key in the sample $S$.  
The key with $p$ value in $(0,1)$ (if there is one) is the new active
key for the cell. 
The storage used is $O(s+|\cL|)$, since we maintain at most one active
key for each part and the number of keys included in $S$ is at most
$s$.
Finally, after completing the pass, we {\sc Pair-Aggregate}
the $\leq |\cL|$ active keys, placing all keys with final $p_i=1$ in $S$.

\ignore{
\medskip
\noindent
{\bf Detailed Description.}
%
In the first pass, we compute 
a \varopt\ (or Poisson IPPS) sample $S'$ (of size  $s'=|S'|$) of the
data set. 
In parallel, compute the IPPS threshold value 
$\tau_s$ to yield a sample of size $s$ (described in Section~\ref{prelim:sec}).
In main memory, we use
$S'$ to construct a partition $\cL$ of the universe of keys. 
Let $L(i)\in\cL$ denote the cell containing key $i$. 
In the second pass, we begin with an empty sample $S$ and no active
keys in each cell $L\in \cL$ (so $a(L) = \emptyset$ for all $L \in \cL$). 
For each key, we apply {\sc IO-aggregate} (Algorithm \ref{io:alg}).
%
%
Lastly, we aggregate all the leftover active keys of different 
cells (there is at most one such key per cell).  

The initial random sample $S'$  guides the
formation of the partition $\cL$. 
 A sample of size $S'=O(s\log s)$ suffices to
hit {\em all} sufficiently large ranges 
(those with combined probability at least $1$) with high probability.
(In geometric terms, it forms an 
$\epsilon$-net of the range space~\cite{epsilonnets:1987}).  
This property enables the sample to form an effective partition
$\cL$.  
The partition has the property that with high probability
$p(L)\leq 1$ for all $L\in\cL$. 
} 

\medskip
\noindent
{\bf Partition and aggregation choices.}
The specifics of the main-memory operations---the construction
of the partition and the final aggregation of active keys---depend on
the range structure.
We start with product structures and then refine the construction to
 obtain stronger results for one-dimensional structures. 
Note that keys in $S'$ with $\min\{1, w_i/\tau_s\}=1$ must be included in $S$.  Moreover, $S'$ must include all such keys---as $S'$ includes all keys with
$w_i \geq \tau_{s'} < \tau_s$, 
 it therefore includes all keys with $w_i\geq \tau_s$.  These keys can thus be 
excluded from consideration after the first phase.

\begin{trivlist}
\item {\em Product structures}: 
We compute $h\leftarrow$ {\sc KD-hierarchy}$(0,S')$ (for $S'$ with 
all keys with $w_i \geq \tau_s$ removed).
This hierarchy $h$ induces a partition of the key domain
according to the splitting criteria in each node.  
The partition $\cL$ corresponds to the leaves of $h$.
The aggregation of active keys in the final phase follows the
hierarchy $h$ (as in Section \ref{drandhierarchy:sec}).

\item {\em Disjoint ranges}: 
There is a cell in the partition for every range from $\cR$
that contains a key from $S'$.
We then induce an arbitrary order over the ranges
and put a cell for each union of ranges in $\cR$ 
which lies between two
consecutive ranges represented in the sample.  
In total we obtain $2s'$ cells.   
In the final phase, active keys can be aggregated arbitrarily.

When $s' =\Omega(s\log s)$,
with probability $1-\poly(1/s)$, all ranges of size $\geq 1$ are
 intersected by $S'$ and each cell $L$ that is a union of ranges not
 seen in $S'$ has size at most $1$ 
(thus, each range $R$ with probability mass $p(R)<1$ 
  can obtain at most one sample in $S$, and so will not be
  over-represented in the final sample).  
Thus, the maximum discrepancy is
  $\Delta<1$ with probability $1-\poly(1/s)$.

\item {\em Order}:
 We sort $S'$ according to the order (excluding keys with
$w_i\geq \tau_s$).  
If $i_1,\ldots,i_t$ are the remaining keys in sorted order, 
there is a cell $L$ for each subinterval
 $(i_j,i_{j+1}]$ and two more, 
 one for all keys $\leq i_1$ and the other for keys  $>i_t$. 
The final aggregation of active keys follows the main-memory
aggregation of ordered keys. 

When $s'=\Omega(s\log s)$, with high probability, the maximum
probability distance between consecutive keys is $1$ and therefore,
the maximum range discrepancy is $\Delta<2$.

\item {\em Hierarchy}:
A natural solution is to linearize the hierarchy, i.e. 
generate a total order consistent with the hierarchy, 
and then apply the algorithm for this order
 structure.
This obtains $\Delta<2$ with high probability. 
Alternatively, we can select all ancestors in the hierarchy of keys
 in $S$ and form a partition by matching each key to its lowest
 selected ancestor.  This will give us maximum range discrepancy
 $\Delta<1$  with high probability.   
The number of ancestors, however, can be large and therefore
this approach is best for shallow hierarchies.  
\end{trivlist}

\newcommand{\incplot}[1]{
\ifpdf
  \includegraphics[width=\figwidth]{#1}
\else
  \includegraphics[angle=270,width=\figwidth]{#1}
\fi
}

\begin{figure*}[t]
\hspace*{-6mm}
\subfigure[Accuracy vs Space on Network]{\incplot{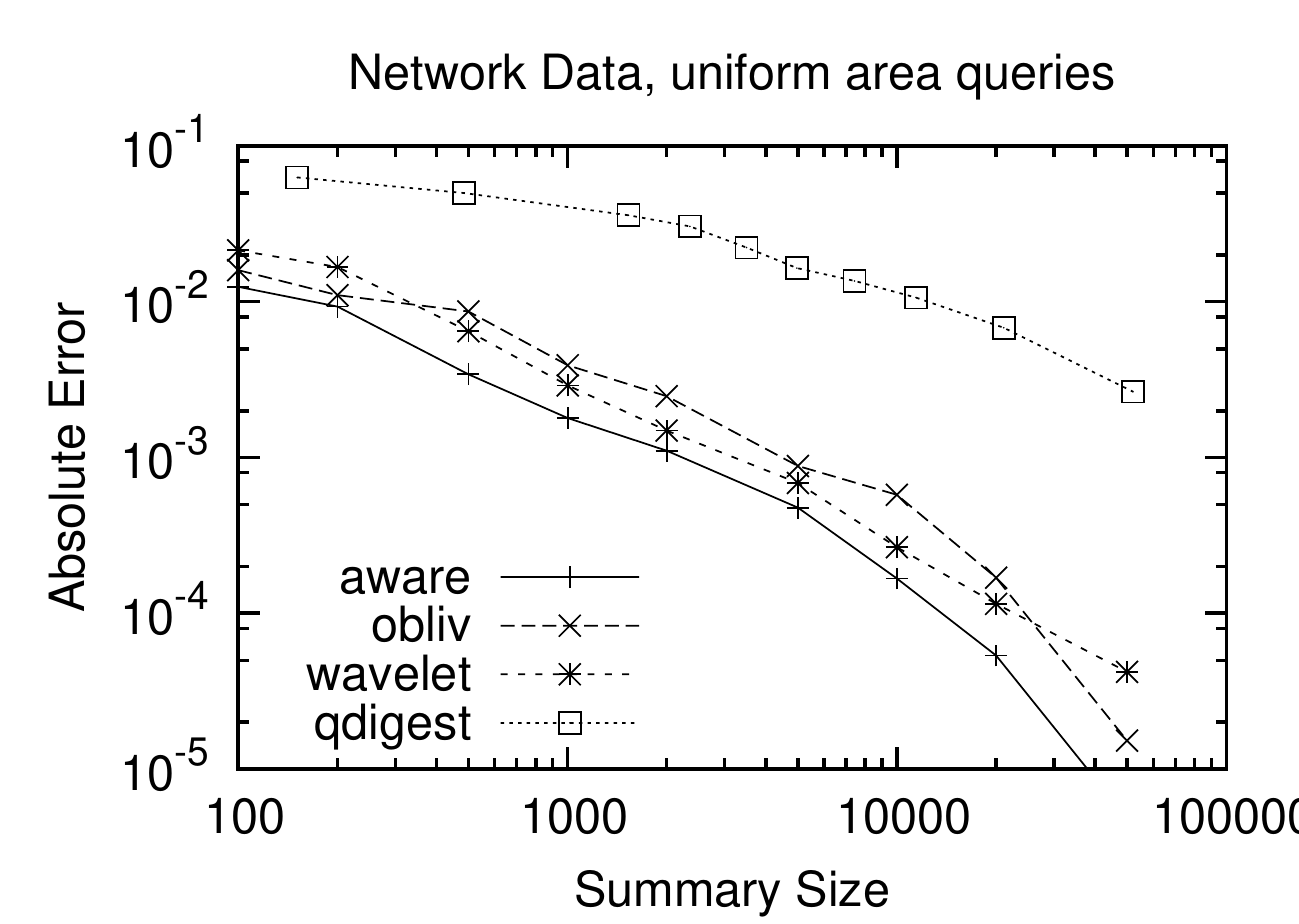}\label{fig:netspace}}%
\subfigure[Accuracy vs Query weight on Network]{\incplot{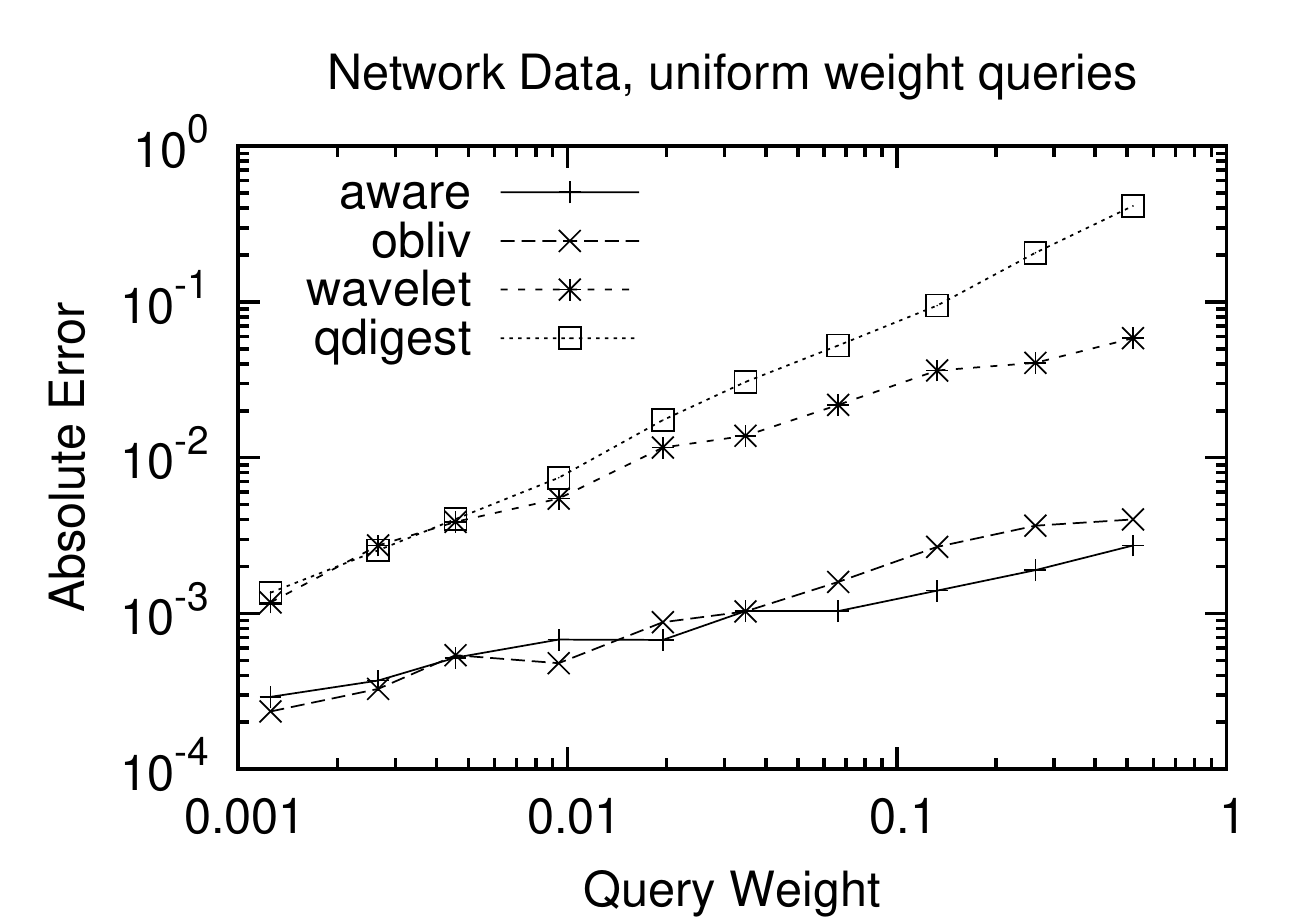}\label{fig:netq}}%
\subfigure[Accuracy vs Ranges per Query on Network]{\incplot{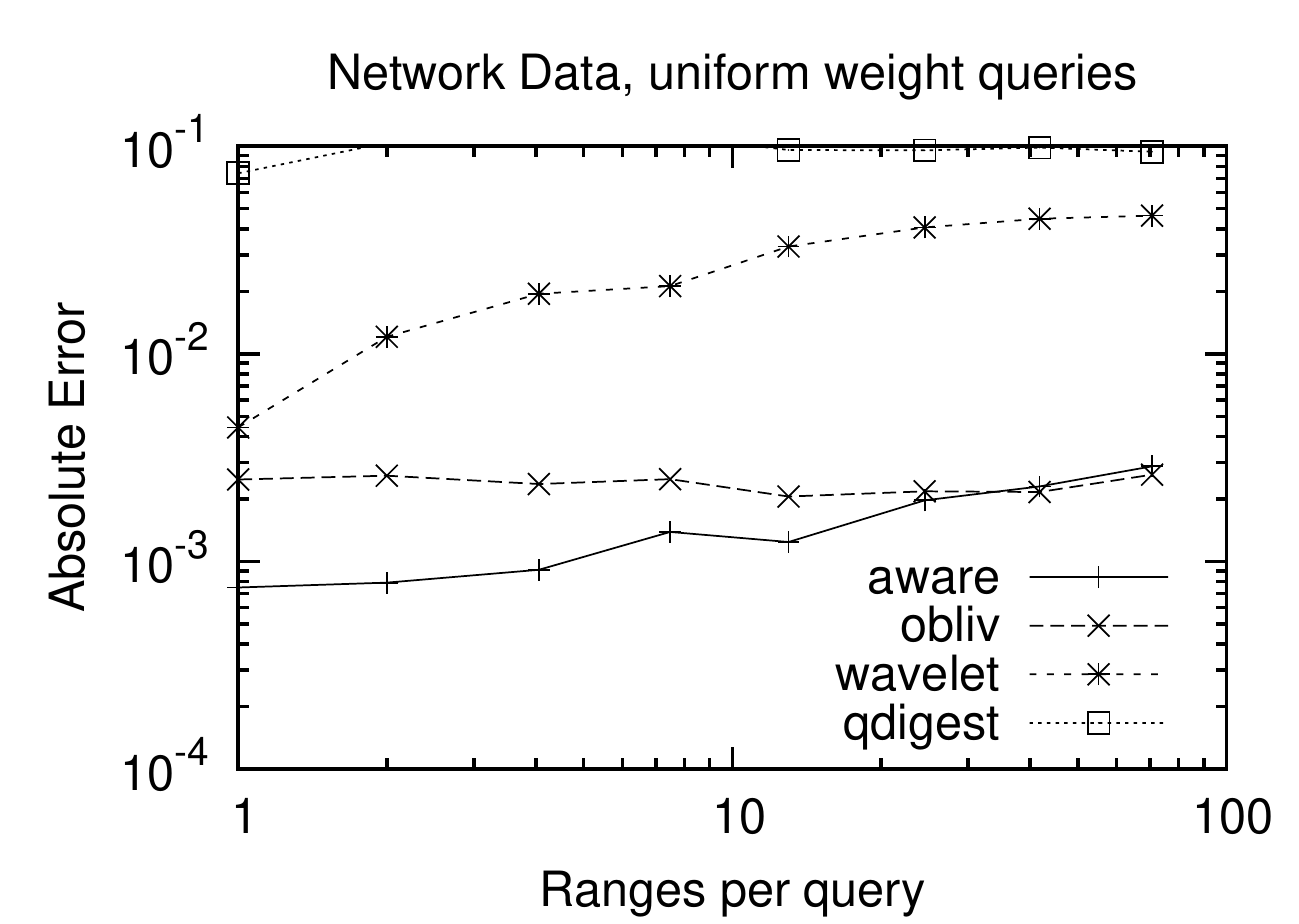}\label{fig:netkd}}
\caption{Experimental results on Network Data set}
\label{fig:network}
\end{figure*}

\section{Experimental Study}
\noindent
We conducted an experimental study of our methods, and compared them
with existing approaches to summarizing large data sets. 
We consider three performance aspects:  building the summary, query
processing, and query accuracy.  We vary the weight of queries and the
number of ranges in each query.  

\ignore{
For structure-aware versus oblivious sampling, the former is easier to build, requiting a single pass over the dataset, query processing is identical, and structure-aware sampling is more accurate.  We expect the accuracy gap to be larger
when the query has fewer ranges and larger weight.  For sampling versus other approaches, sample-based summaries are generally more efficeint to build and use.
They also have stronger advantage when there are more ranges per query.
}

\subsection{Experimental Environment}

\para{Data Sets and Query Sets.}
We compared our approach on a variety of data sets,
and present results on two examples:
\begin{trivlist}
\item The {\em Network} dataset consists of network flow data summarizing
  traffic volumes exchanged between a large number of sources and
  destination, observed at a network peering point.
In total, there are 63K sources and 50K destinations, and a total of
196K pairs active in the data.
The space of these records is determined by the two-dimensional IP
hierarchy, i.e. X=$2^{32}$ and Y=$2^{32}$.

\item
The {\em Technical Ticket} data is derived from calls to a customer
care center of a broadband wireline access network that resulted in a
technical trouble ticket. 
Each key consists of: (i) an anonymous identifier for unique
customers; (ii) a trouble code, representing a point in a
hierarchy determining the nature of the problem identified from a
predetermined set by the customer care staff; and (iii) a network code,
indicating points on the network path to the customer's location.  
Both attributes are hierarchical with varying branching factor at each
level, representing a total of approximately $2^{24}$ possibilities
in each dimension, i.e. $X=2^{24}$ and $Y=2^{24}$. 
There are 4.8K distinct trouble codes present in the data, 
80K distinct network
locations, and 500K distinct observed combinations. 
\end{trivlist}

Each query is produced as a collection of non-overlapping rectangles
in the data space.
To study the behavior of different summaries over
different conditions, we generated a variety of queries of two types. 
In the {\em uniform area} case, 
each rectangle is placed randomly, with height and width chosen
uniformly within a range $[0, h]$ and $[0,w]$. 
We tested a variety of scales to determine $h$ and $w$, varying these
from covering almost the whole space, to covering only a very small
fraction of the data (down to a $10^{-4}$ fraction). 
In the {\em uniform weight} case, 
each rectangle is chosen to cover roughly the same fraction of the
total weight of keys.
We implement this by building a kd-tree over the whole data, and
picking cells from the same level (note, this is independent of any
kd-tree built over sampled data by our sampling methods). 
For each query, we compute the exact range sum over the data, and
compare the absolute, sum-squared and relative errors of our methods
across a collection of queries. 
In our plots below, we show results from a battery of 50 queries with
varying number of rectangles.

\para{Methods.}
We compared  structure-aware sampling  to examples of various
classes of alternatives, as described in Appendix \ref{prelim:sec}:

\begin{trivlist}
\item
{\em Obliv}, is a structure-oblivious sampling method. 
We implemented \varopt\ sampling
to give a sample size of exactly $s$.

\item
 {\em Aware}, the structure aware sampling method. 
 We follow the 2 pass algorithm (Section \ref{IO:sec} and
Section \ref{product:sec}), and first
 draw a sample of size several times larger than $s$, the target
 size (in our experiments, we set $s'= 5s$: increasing the
 factor did not significantly improve the accuracy).
 We then built the kd-tree on this sample, and took a second pass over
 the data to populate the tree.  
 Lastly, we perform pair aggregation within the tree structure to
 produce a sample of size exactly $s$. 
Although the algorithm is more involved than straight \varopt, 
 both are implemented in fewer than 200 lines of code.

\item
{\em Wavelet}, implements the (two-dimensional) standard Haar wavelet scheme.  
In a single pass we compute the full wavelet transform of the data:
each input data point contributes to $\log X \cdot \log Y$ wavelet
coefficients. 
We then prune these coefficients to retain only the $s$ largest
(normalized) coefficients for query answering. 
\item
{\em Qdigest}, implements the (two-dimensional) q-digest data
structure \cite{HSST:ISAAC04}.
This deterministically builds a summary of the data. 
Given a parameter $\epsilon$, the structure is promised to be at worst 
$O(\frac{1}{\epsilon^2}\log X \cdot \log Y)$, but in practice
materializes much fewer nodes than this, so we count the number of
materialized nodes as its size.
\item
{\em Sketch}, implements the Count-sketch, a randomized summary of the
data \cite{ccf:icalp2002}. 
Each input item contributes to $O(\log X \cdot \log Y)$ sketches of
dyadic rectangles.  
We choose the size of each sketch so that the total size is bounded by 
a parameter $s$. 
\end{trivlist}

Our implementations in Python were run on the same machine, on a 2.4GHz core
with access to 4GB of memory. 
For most methods, we perform static memory allocation as a function of
summary size in advance. 
The exception is wavelets, which needs a lot of space to build the
transform before pruning. 

\subsection{Network Data Accuracy}

\begin{figure*}[t]
\hspace*{-6mm}
\subfigure[Construction throughput: network data]{
\incplot{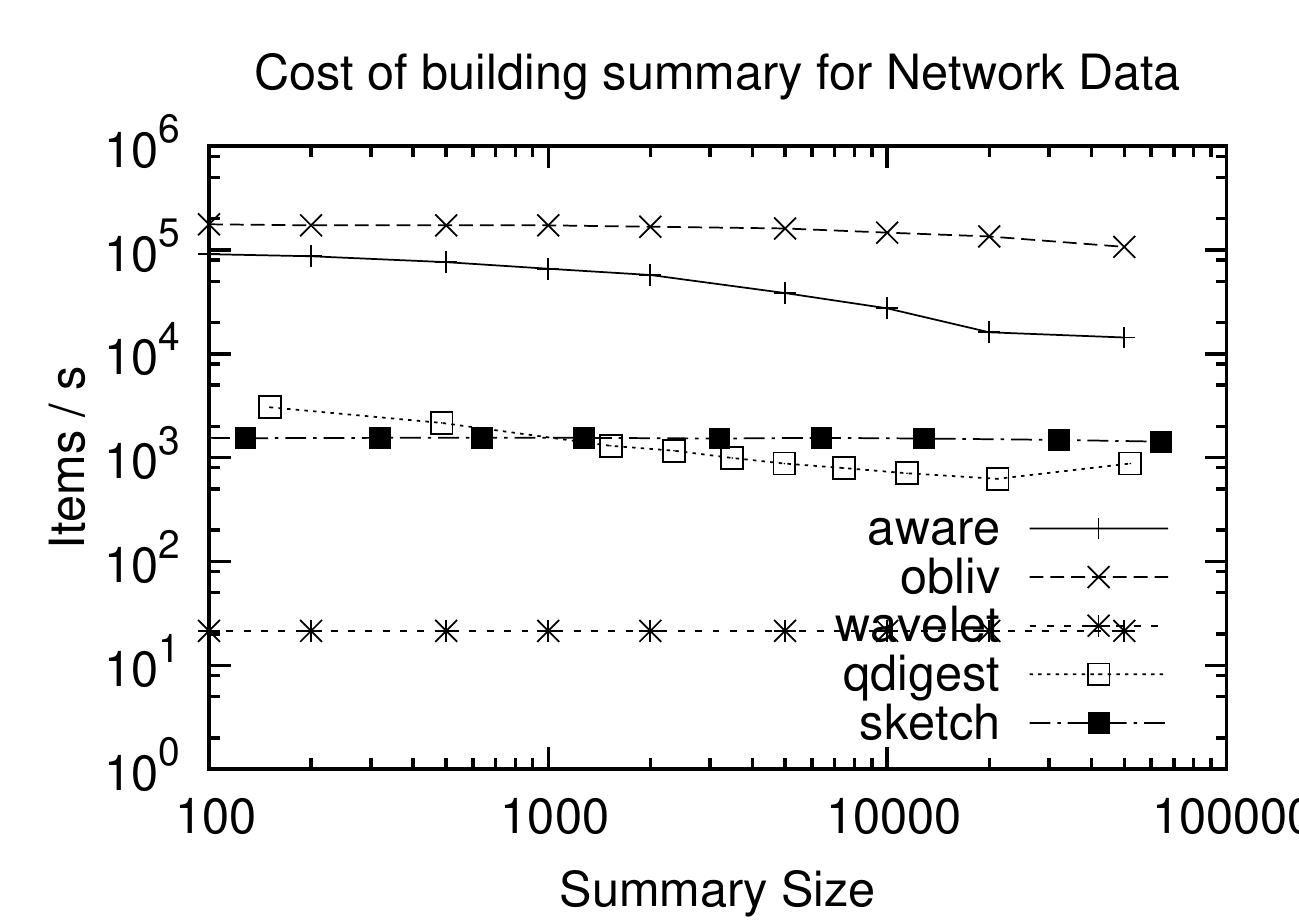}
\label{fig:buildtime}
}%
\subfigure[Construction throughput: tech ticket data]{
\incplot{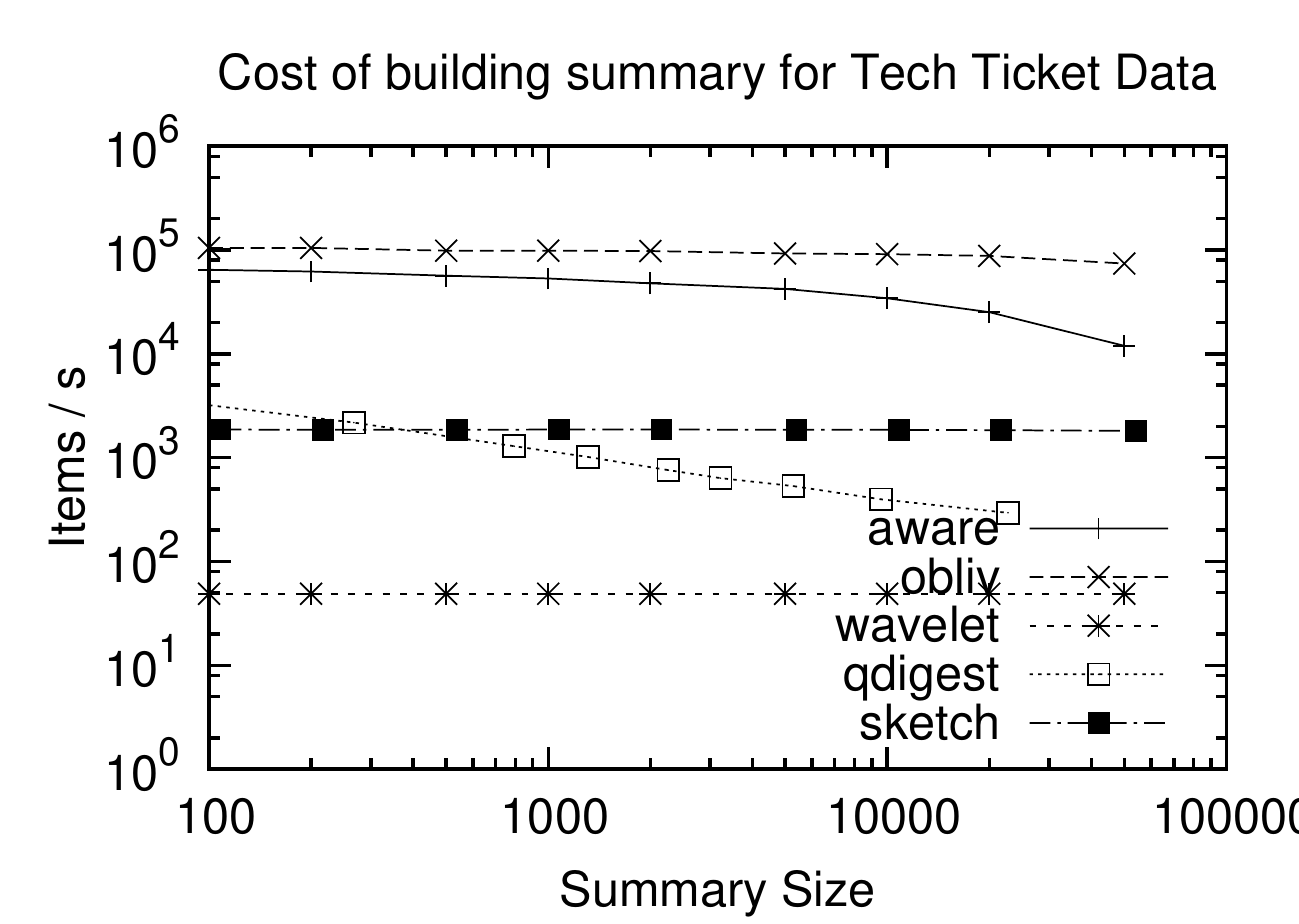}
\label{fig:techthru}
}
\subfigure[Query time vs summary size]{
\incplot{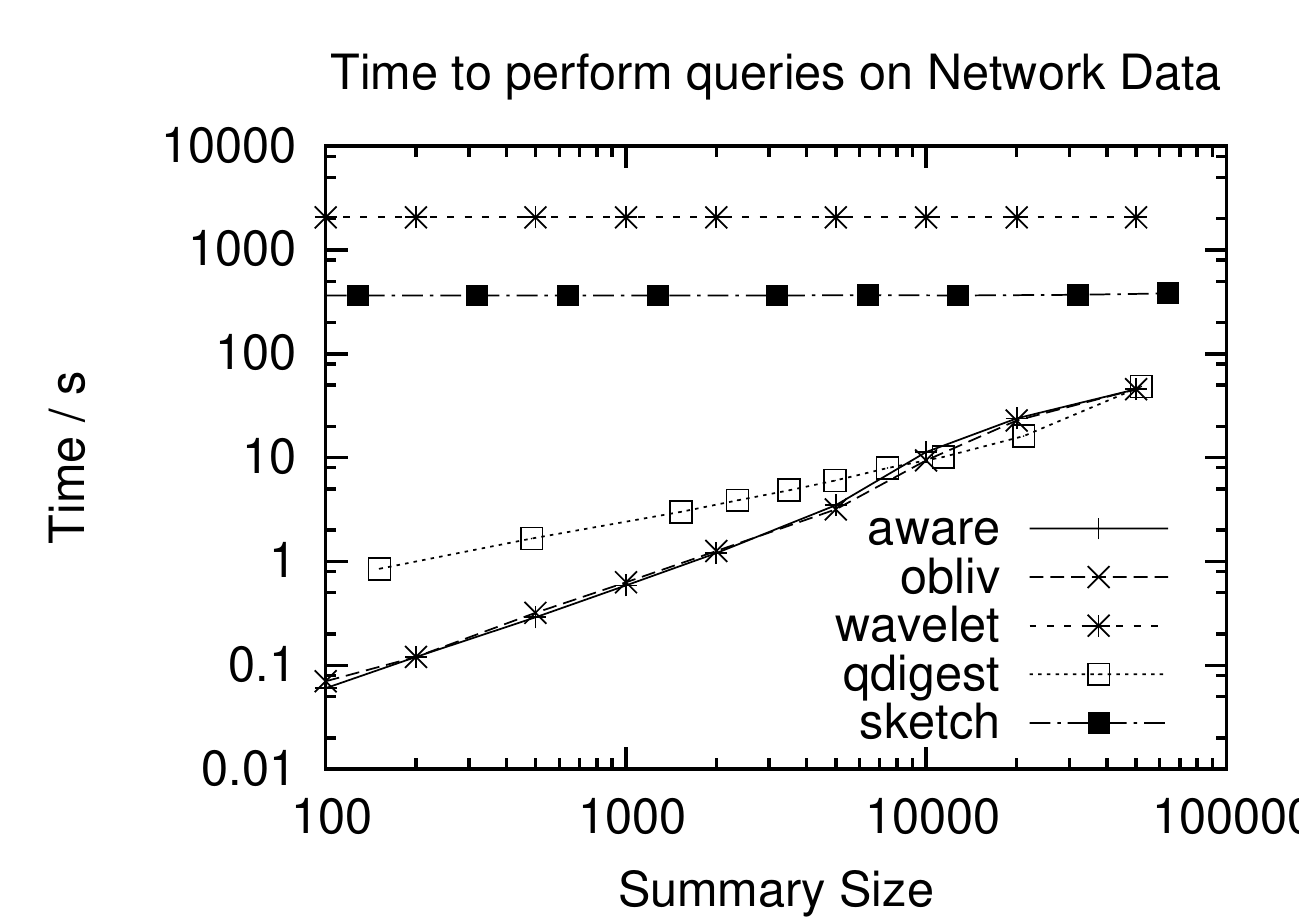}
\label{fig:querytime}
}
\caption{Time costs as summary size varies}
\label{fig:time}
\end{figure*}

Figure \ref{fig:network} shows accuracy results on network data.
The y-axis shows accuracy measured as the error in the query results
divided by the total weight of all data (the absolute error). 
Our experiments which computed other metrics such as 
sum-squared error showed the same trends, and so are omitted. 

On this data, the structure-aware sampling typically achieved the
least error.
Figure \ref{fig:netspace} shows this behavior across a variety of
summary sizes  with uniform area queries each containing 25 ranges.  
For comparison, we measure the space used by each summary in terms of
elements on the original data, so in this case the smallest sample
contains 100 IP address pairs (and weights).
This is compared to keeping the 100 largest wavelet coefficients, and
a q-digest of 100 nodes. 
We also compared to sketch summaries with an equivalent amount of
space. 
However, the space required before a sketch of two-dimensional data
becomes accurate is
much larger than for the other summaries considered. 
The total error for the queries shown was off the scale on the graphs,
so we omit sketches from further accuracy comparison. 

Across most summary sizes, the structure-aware sampling is several
times more accurate than its structure-oblivious counterpart:
 Figure~\ref{fig:netspace}, which is in log-scale on both axes,
shows that the error of the structure-aware method is half to a third
that of the oblivious method given the same space: a significant
improvement.  
The deterministic q-digest structure is one to two orders of magnitude
less accurate in the same space. 
Only the Haar wavelet comes close to structure-aware sampling in terms
of accuracy. 
This is partly due to the nature of uniform area queries:
on this data, these correspond to either very small or very large weight.

Figure \ref{fig:netq} shows the accuracy under uniform weight queries, 
where each query contains 10 ranges of approximately equal weight.  
The graph shows the results for a fixed summary size of 2700 keys 
(about 32KB).
Here we see a clear benefit of sampling methods compared to wavelets
and q-digest:  note that the error of q-digest is close to
the total weight of the query. 
For ``light'' queries, comprising a small fraction of the data, there
is little to choose between the two sampling methods.  
But for queries which touch more of the data, structure-awareness
obtains half the error of its oblivious cousin. 
The general trend is for the absolute error to increase with the
fraction of data included in the sample. 
However, 
the gradient of these lines is shallow, meaning that
 the {\em relative} error is actually improving, as our analysis indicates: 
 structure aware sampling obtains 0.001 error on a query that covers
more than 0.1 of the whole data weight, i.e., the observed
relative error is less than 1\%. 

Figure \ref{fig:netkd} shows the case when we hold the total weight of
the query steady (in this case, at approximately 0.12 of the total
data weight), and vary the number of ranges in each query. 
We see that the structure oblivious accuracy does not vary much: 
as far as the sample is concerned, all the queries are essentially
subset queries with similar weight. 
However, when the query has fewer ranges, the structure aware sampling can
be several times better. 
As the number of ranges per query increase, each range becomes
correspondingly smaller, and so for queries with 40 or more ranges
there is minimal difference between the structure aware and structure
oblivious. 
On this query set, wavelets are an order of magnitude less accurate,
given the same summary size. 

\subsection{Scalability}
To study scalability of the different methods, we measured the time to
build the summaries, and the time to answer 2500
range queries. 
Figure \ref{fig:time} shows these results as we
vary the size of the summary. 
Although our implementations are not
optimized for performance, we believe that these results show the general
trend of the costs. 
Structure-oblivious is the fastest, whose cost is essentially
bounded by the time to take a pass through the data
(Figures \ref{fig:buildtime} and \ref{fig:techthru}). 
Structure-aware sampling requires a second pass through the data, 
and for large summaries the 
extra time to locate nodes in the kd-tree reduces the throughput.  
We expect that a more engineered implementation could reduce this
building cost. 

\begin{figure*}[t]
\hspace*{-6mm}
\subfigure[Accuracy vs space on Tech Ticket]{\incplot{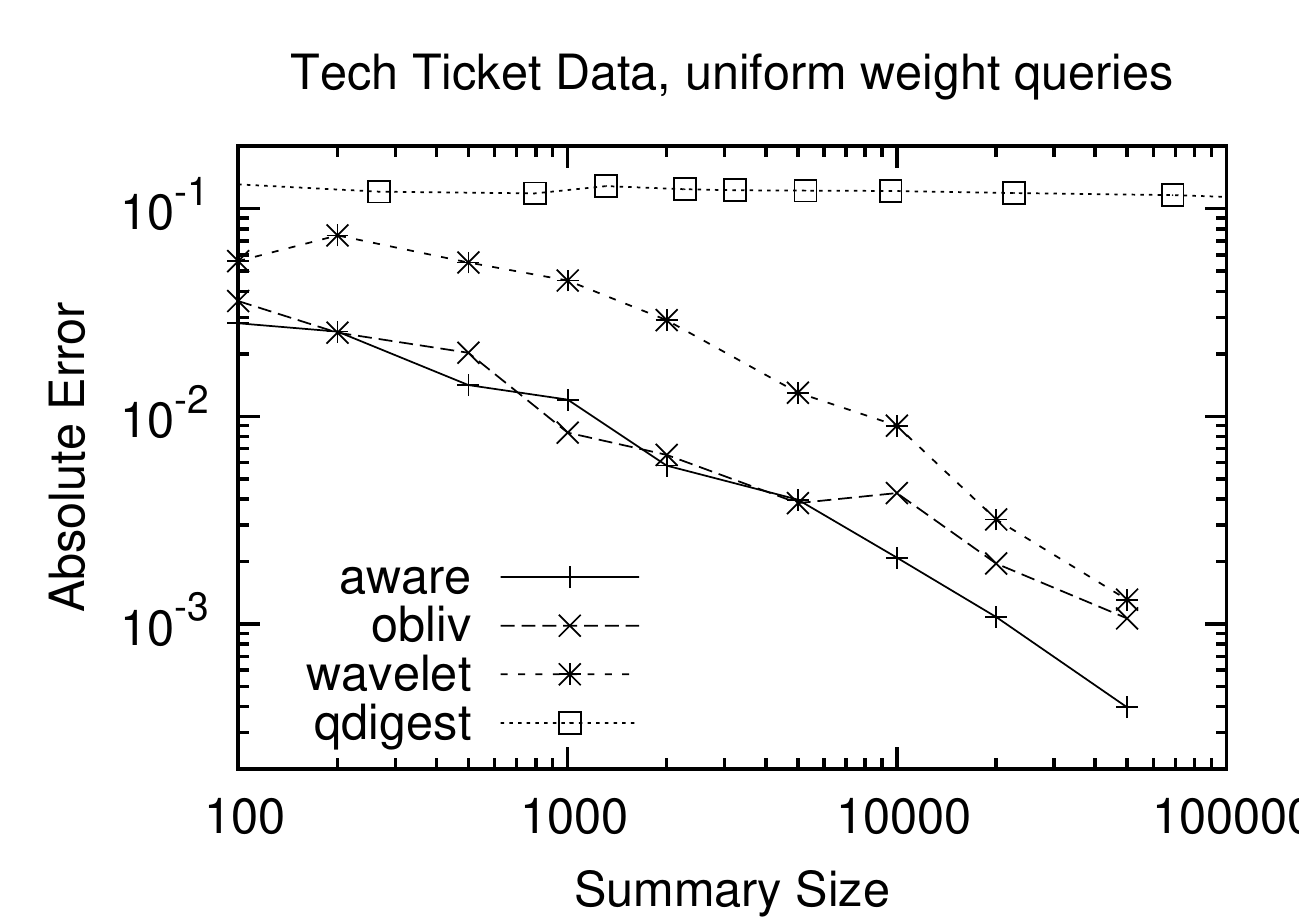}\label{fig:techspace}}%
\subfigure[Accuracy vs query weight on Tech Ticket]{\incplot{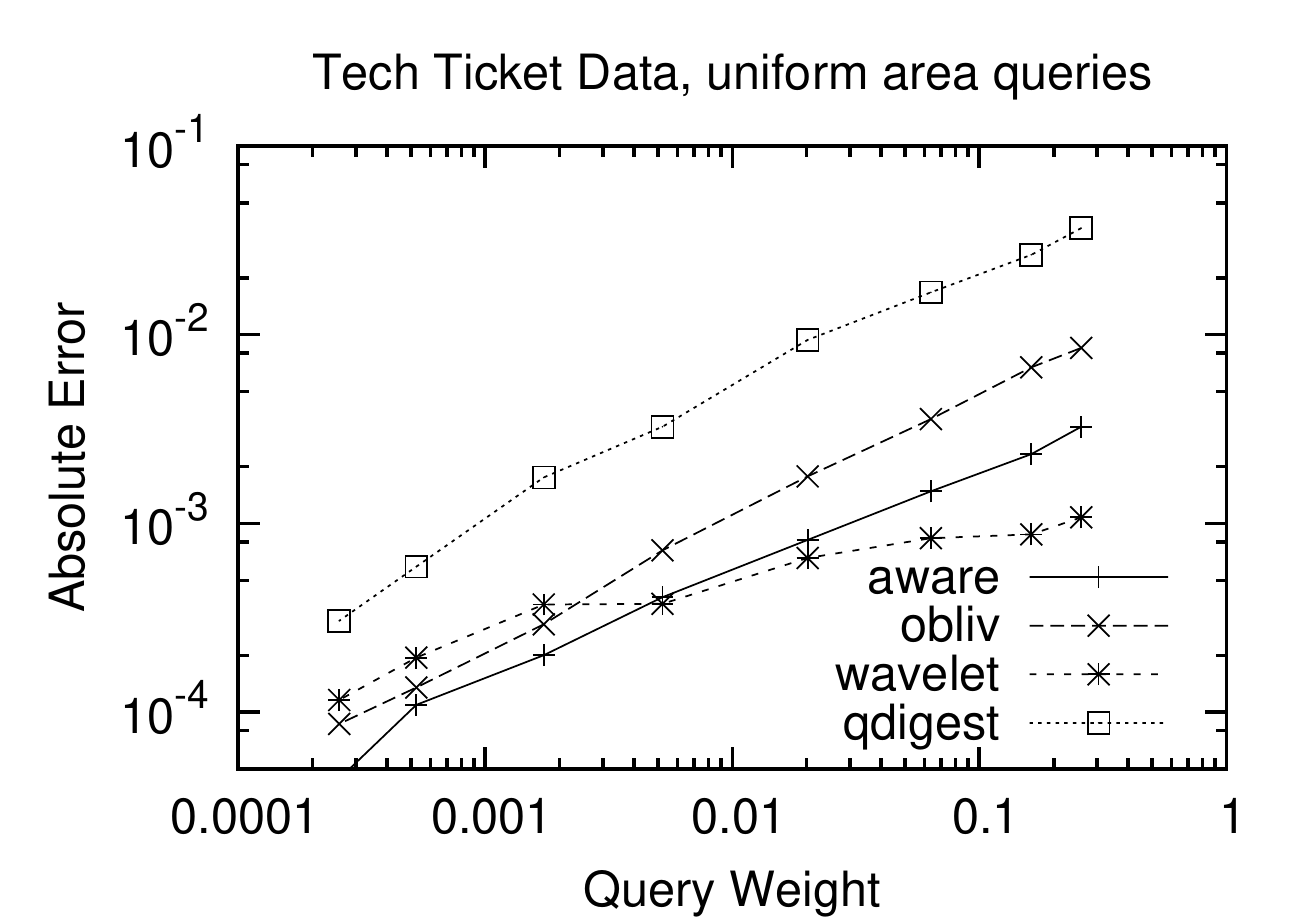}\label{fig:techq}}%
\subfigure[Accuracy vs query weight on Tech Ticket]{\incplot{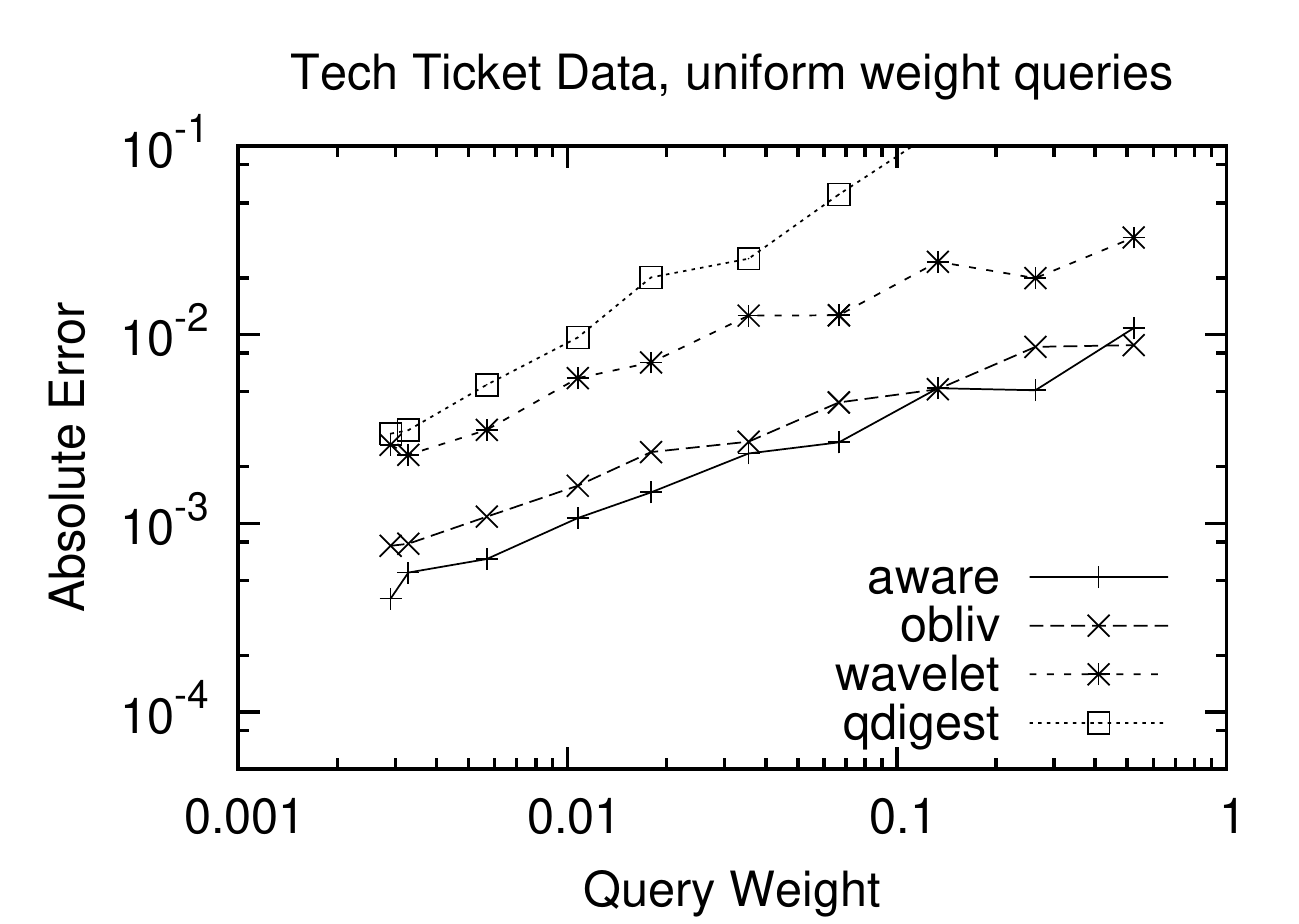}\label{fig:techk}}
\caption{Accuracy on the Tech Ticket Data}
\label{fig:tech}
\end{figure*}

The q-digest and sketch summaries are both  around 2 orders of
magnitude slower to build the summary. 
These structures are quite fast in one-dimension, but have to do more
work in higher dimensions. 
For example, the sketch needs to update a number of locations which
grows with the product of the logarithm of the dimension sizes. 
On pairs of 32 bit addresses, this factor is proportional to
$32^2$=1024. 
The cost of building the 2D Haar wavelet summary is 4 orders of magnitude
more than sampling. 
The reason is that each point in the data contributes to 1024
coefficients, leading to millions of values before thresholding. 

Since the samples, once built, have the same form, query answering
takes the same time with both obliv and aware (Figure
\ref{fig:querytime}): 
we just compute the
intersection of the sample with each query rectangle. 
The cost grows with increasing sample size, as we are just scanning
through the sample to find which points fall within each rectangle. 
Still, this naive approach to query answering can process thousands of
query rectangles per second (recall, the y-axis is showing the time to
complete all 2500 rectangle queries). 
In comparison, asking this many queries over the full data takes 2 minutes.
Again, we see the disadvantage of the wavelet approach: each rectangle
query takes of the order of a second---in other words, it is about
1000 times slower than using a sample. 
The reason is that each rectangle is decomposed into
dyadic rectangles.
In the worst case, there are 1000 dyadic rectangles formed, and each
dyadic rectangle requires  the value of 1000 coefficients. 
The effect as we go to higher dimensions only gets worse, growing
exponentially with the dimension.
While there should be more efficient ways to use wavelets, the overall
cost is offputtingly high. 

\subsection{Tech Ticket Data Accuracy}

The plots in Figure \ref{fig:tech} show accuracy experiments
on the tech ticket data set. 
Figure \ref{fig:techspace} shows that
structure-aware and structure-oblivious sampling behave similarly for
smaller sample sizes: this is because this data set has many high
weight keys which must be included in both samples.
For large sample sizes, the methods diverge, and the benefit of
structure awareness is seen: the error is less than half that for the
same sized obliv summary for samples that are between 1\% and
10\% of the data size.  


The next two plots compare the case for uniform area queries 
(over 25 ranges, Figure \ref{fig:techq}) and uniform weight queries
(10 ranges, Figure \ref{fig:techk}).
We see that on uniform area queries, wavelets can become competitive
for higher weights, but this is not seen when the weight of each range
is controlled. 
In either case, for these queries of several ranges, structure-aware
sampling seems to give the best results overall . 
Certainly, wavelets do not seem to scale with this
type of data: tens to hundreds of millions of coefficients are
generated before thresholding, leading to a high time and space cost. 
Figure~\ref{fig:techthru} emphasises the impracticality of wavelets on
this data: generating and using samples takes seconds,
while using wavelets takes  (literally) hours.


\section{Concluding Remarks}

We introduce structure-aware sampling as an alternative to
structure-oblivious sampling and tailored deterministic summaries. 
Our structure-aware samples are \varopt---they retain the full benefits
of state-of-the-art sampling over deterministic summaries: flexibility
and support for arbitrary subset queries, accuracy on these queries,
unbiased estimation, and exponential tail bounds on the error.  
By optimizing the sample distribution for range queries, we obtain
superior accuracy with respect to structure-oblivious samples and  
match or surpass the accuracy of tailored deterministic summaries.
Our proposed algorithms are simple to implement and are I/O
efficient, requiring only two sequential read passes over the data
and memory which is independent of the input size. 
Going to only a single (streaming) pass requires quite different
ideas, since in this case the \varopt\ sample is unique, and hence
structure oblivious. 
Instead, it is necessary to relax the \varopt\ requirement to allow
the freedom to exploit structure; our initial results in this
direction are presented in \cite{CCD:sigmetrics11}.


\begin{thebibliography}{10}
\addtolength{\itemsep}{-0.33mm}
\bibitem{BCEG:ToA2007}
A.~Bagchi, A.~Chaudhary, D.~Eppstein, and M.~T. Goodrich.
\newblock Deterministic sampling and range counting in geometric data streams.
\newblock {\em ACM Trans. Algorithms}, 3(2):16, 2007.

\bibitem{BCDHS:KDD2003}
H.~Br\"{o}nnimann, B.~Chen, M.~Dash, P.~Haas, and P.~Scheuer- mann.
\newblock Efficient data reduction with {EASE}.
\newblock In {\em KDD}, 2003.

\bibitem{Cha82}
M.~T. Chao.
\newblock A general purpose unequal probability sampling plan.
\newblock {\em Biometrika}, 69(3):653--656, 1982.

\bibitem{ccf:icalp2002}
M.~Charikar, K.~Chen, and M.~Farach-Colton.
\newblock Finding frequent items in data streams.
\newblock In {\em  ICALP}, 2002.

\bibitem{Chernoff}
H.~Chernoff.
\newblock A measure of the asymptotic efficiency for test of a hypothesis based
  on the sum of observations.
\newblock {\em Annals of Math. Statistics}, 23:493--509, 1952.

\bibitem{CCD:sigmetrics11}
E.~Cohen, G.~Cormode, N.~Duffield. 
\newblock Structure-Aware Sampling on Data Streams
\newblock To appear in {\em SIGMETRICS}, 2011

\bibitem{varopt:CDKLT08}
E.~Cohen, N.~Duffield, H.~Kaplan, C.~Lund, and M.~Thorup.
\newblock Stream sampling for variance-optimal estimation of subset sums.
\newblock In {\em SODA}, 2009.

\bibitem{varopt_full:CDKLT10}
E.~Cohen, N.~Duffield, C.~Lund, M.~Thorup, and H.~Kaplan.
\newblock Variance optimal sampling based estimation of subset sums.
\newblock arXiv:0803.0473v2, 2010.

\bibitem{HHH:CKMS:TKDD2008}
G.~Cormode, F.~Korn, S.~Muthukrishnan, and D.~Srivastava.
\newblock Finding hierarchical heavy hitters in streaming data.
\newblock {\em ACM Trans. Knowl. Discov. Data}, 1(4):1--48, 2008.

\bibitem{GandhiKPS:jacm06}
R.~Gandhi, S.~Khuller, S.~Parthasarathy, and A.~Srinivasan.
\newblock Dependent rounding and its applications to approximation algorithms.
\newblock {\em J. Assoc. Comput. Mach.}, 53(3):324--360, 2006.

\bibitem{GKTD:sigmod2000}
D.~Gunopulos, G.~Kollios, V.~J. Tsotras, and C.~Domeniconi.
\newblock Approximating multi-dimensional aggregate range queries over real
  attributes.
\newblock In {\em SIGMOD}, 2000.

\bibitem{Hajekbook1981}
J.~H{\'a}jek.
\newblock {\em Sampling from a finite population}.
\newblock Marcel Dekker, New York, 1981.

\bibitem{epsilonnets:1987}
D.~Haussler and E.~Welzl.
\newblock Epsilon nets and simplex range queries.
\newblock {\em Discrete Comput. Geom.}, 2, 1987.

\bibitem{HSST:ISAAC04}
J.~Hershberger, N.~Shrivastava, S.~Suri, and C.~D. T{\'o}th.
\newblock Adaptive spatial partitioning for multidimensional data streams.
\newblock In {\em ISAAC}, 2004.

\bibitem{HT52}
D.~G. Horvitz and D.~J. Thompson.
\newblock A generalization of sampling without replacement from a finite
  universe.
\newblock {\em J. Amer. Stat. Assoc.}, 47(260):663--685, 1952.

\bibitem{LKC:sigmod1999}
J.-H. Lee, D.-H. Kim, and C.-W. Chung.
\newblock Multi-dimensional selectivity estimation using compressed histogram
  information.
\newblock {\em SIGMOD Rec.}, 28(2):205--214, 1999.

\bibitem{Matias98wavelet-basedhistograms}
Y.~Matias, J.~S. Vitter, and M.~Wang.
\newblock Wavelet-based histograms for selectivity estimation.
\newblock In {\em SIGMOD}, 1998.

\bibitem{PancSri:sicomp97}
A.~Panconesi and A.~Srinivasan.
\newblock Randomized distributed edge coloring via an extension of the
  chernoff-hoeffding bounds.
\newblock {\em SIAM J. Comput.}, 26(2):350--368, 1997.

\bibitem{Phillips:icalp2008}
J.~M. Phillips.
\newblock Algorithms for epsilon-approximations of terrains.
\newblock In {\em ICALP }, 2008.

\bibitem{PoosalaIoannidis:VLDB97}
V.~Poosala and Y.~Ioannidis.
\newblock Selectivity estimation without the attribute value independence
  assumption.
\newblock In {\em VLDB},  1997.

\bibitem{SSW92}
C.-E. S\"{a}rndal, B.~Swensson, and J.~Wretman.
\newblock {\em Model Assisted Survey Sampling}.
\newblock Springer, 1992.

\bibitem{ShrivastavaBAS:sensys04}
N.~Shrivastava, C.~Buragohain, D.~Agrawal, and S.~Suri.
\newblock Medians and beyond: new aggregation techniques for sensor networks.
\newblock In {\em SenSys}, 2004.

\bibitem{Sri01}
A.~Srinivasan.
\newblock Distributions on level-sets with applications to approximation
  algorithms.
\newblock In {\em FOCS}, 2001.

\bibitem{SuriTZ:CG2004}
S.~Suri, C.~D. T\'{o}th, and Y.~Zhou.
\newblock Range counting over multidimensional data streams.
\newblock In {\em SCG}, 2004.

\bibitem{ST07}
M.~Szegedy and M.~Thorup.
\newblock On the variance of subset sum estimation.
\newblock In {\em  ESA}, 2007.

\bibitem{Tille:book}
Y.~Till\'e.
\newblock {\em Sampling Algorithms}.
\newblock Springer, 2006.

\bibitem{VCepsapprox:1971}
V.~Vapnik and A.~Y. Chervonenkis.
\newblock On the uniform convergence of relative frequencies of events to their
  probabilites.
\newblock {\em Theory of Probability and its applications}, 16:264--280, 1971.

\bibitem{Vitter98datacube}
J.~S. Vitter, M.~Wang, and B.~Iyer.
\newblock Data cube approximation and histograms via wavelets.
\newblock In {\em CIKM}, 1998.

\bibitem{ZSSDC:IMC04}
Y.~Zhang, S.~Singh, S.~Sen, N.~Duffield, and C.~Lund.
\newblock Online identification of hierarchical heavy hitters: algorithms,
  evaluation, and applications.
\newblock In {\em IMC}, 2004.

\end{thebibliography}
{
}

\newpage
\appendix

\section{Background on Summarization} \label{prelim:sec}

\subsection*{Sampling Techniques}
A sample $S$ is a random subset drawn from the space of keys $K$. 
In our context, the keys $K$ are members of a structured domain $\cK$.

\noindent
{\bf Inclusion Probability Proportional to Size (IPPS)~\cite{Hajekbook1981}:}
Many sampling schemes set the inclusion
probability of each key in $S$  proportional to its weight, but truncated
so as not to exceed $1$.  
When defined with respect to a {\em threshold} 
parameter $\tau >0$, the inclusion probability 
of $i$ is $p_i=\min\{1,w_i/\tau\}$.  

The expected size of the sample is just the sum of the $p_i$'s, so we
can achieve a sample of expected size $s$ by choosing an appropriate
threshold $\tau_s$. 
This $\tau_s$ can be found by solving the equation: 

\smallskip
\centerline{
$\sum_i \min\{1, w_i/\tau_s\} =s\ .$
}
\smallskip
\noindent
 $\tau_s$ can be computed via a linear pass on the data, using a
heap $H$ of size at most $s$. 
Let $L$ denote the total weight of keys processed that are not present
in the heap $H$. 
Algorithm \ref{get_tau_k:alg} presents the algorithm to update the
value of $\tau$ for each new key. 
If the weight of the new key is below the current value of $\tau$,
then we add it to $L$, the sum of weights; else, we include the weight
in the heap $H$. 
Then adjust the heap: if the heap has $s$ items, or the smallest
weight in the heap falls below the new value of $\tau$
(Line~\ref{linewhile}), we move the smallest weight from $H$ to
$L$. 
Finally, we compute the new $\tau$ value in Line~\ref{line_get_tau}. 
After processing all keys, we have found $\tau_s$ for this data.

\ignore{
{\bf Initialize:}
\begin{algorithmic}[1]
\State $H\leftarrow \emptyset$ \Comment{Initialize an empty heap}
\State $L\leftarrow 0$ \Comment{Total weight of processed items that are not in $H$}
\State  $\tau\leftarrow 0$
\end{algorithmic}
}

\medskip
\noindent
{\bf The Horvitz Thompson (\HT) estimator~\cite{HT52}} 
allows us to accurately answer queries from a sample, by 
providing {\em adjusted weights} to use for each key in the sample. 
For a key $i$ of weight $w_i$ and inclusion probability $p_i$ the 
adjusted weight is $a(i)=w_i/p_i$ if $i$ is included in
the sample and $a(i)\equiv 0$ otherwise.  
For all $i$, $a(i)$ is an optimal estimator of $w_i$ in that it minimizes 
the variance $\var[a(i)]$.  

 A summary that includes the keys in $S$ and 
their \HT\ adjusted weights can be used to 
estimate the weight $w(J)$ of any subset of keys $J\subset K$: 
 $a(J)=\sum_{i\in J} a(i)=\sum_{i\in {S\cap J}} a(i)$.
 The \HT\ estimates are clearly
unbiased: for all $i$, $\E[a(i)]=w_i$ and from linearity of expectation,
for all subsets $J$, $\E[a(J)]=w(J)$.  

With IPPS, the \HT\ adjusted weight of an included key
is $\tau$ if $w_i\leq \tau $ and $w_i$ otherwise.  
Hence, for any subset $J$ we have
\begin{equation} \label{ipps_est_subset}
\textstyle
a(J)=\sum_{i\in J | w_i\geq \tau} w_i + \left\vert\{ i\in J\cap S : w_i<\tau \}\right| \cdot \tau\ .
\end{equation}
We can store either the adjusted weights for
each key, or the original weights (and compute the adjusted weights
via $\tau$ on demand). 
The variance of adjusted weight $a_i$ 
is $\var[a_i]=w_i^2(1/p_i-1)\equiv w_i(\tau-w_i)$ 
if $w_i\leq \tau$ and $0$ otherwise. 
Using IPPS inclusion probabilities with \HT\ estimates
minimizes the sum $\SV[a]=\sum_i \var[a(i)]$ of per-key variances 
over all inclusion probabilities and estimators with
the same (expected) sample size $s=\sum_i p_i$.

\medskip
\noindent
{\bf Poisson sampling} is where
the choice of whether to sample any key is made independently of all
other choices.
A Poisson IPPS sample of a specified expected size $s$ can be
 computed efficiently with a single pass over the data.

\begin{algorithm}[t]
\caption{{\sc Stream\_$\tau$}$(i)$ : processing item $i$}\label{get_tau_k:alg}
\begin{algorithmic}[1]
\If {$w_i<\tau$} $L\leftarrow L + w_i$
\Else \, Insert$((i,w_i),H)$;
\EndIf
\While {($|H|=s$)\ {\bf or}\ (min$(H) < \tau$)} \label{linewhile}
\State $a \leftarrow $ delete\_min$(H)$.
\State $L \leftarrow  L+w_a$
\State $\tau \leftarrow \frac{L}{s-|H|}$ \label{line_get_tau}
\EndWhile
\end{algorithmic}
\end{algorithm}

\medskip
\noindent
{\bf \varopt\ sampling}~\cite{Cha82,Tille:book,varopt:CDKLT08}
 improves over Poisson sampling by guaranteeing a fixed
sample size (exactly $s$) and giving tighter
estimates~\cite{varopt:CDKLT08,varopt_full:CDKLT10,Cha82,Tille:book,GandhiKPS:jacm06}: 
the variance of any subset-sum
 estimate is at most that obtained by a Poisson IPPS sample.
\varopt\ samples are optimal in that 
for {\em any} subset size, they minimize the 
average variance of the estimates~\cite{ST07,varopt:CDKLT08}.
A \varopt\ sample of size $s$, denoted \varopt$_s$, 
can be computed with a single pass over the
data~\cite{varopt_full:CDKLT10}, 
generalizing the notion of reservoir sampling. 

\noindent
A sample distribution over $[n]$ is \varopt\ for a parameter $\tau$ if:
 \begin{trivlist}
\item [(i)]
  Inclusion probabilities are IPPS, i.e. for a key $i$,
$p_i=\min\{1,w_i/\tau\}$. 
\item [(ii)]
 The sample size is exactly $s=\sum_{i\in [n]} p_i$. 
 \item [(iii)] 
{\em High-order inclusion \& exclusion probabilities are bounded 
 by respective products of first-order probabilities}, so, 
for any $J\!\subseteq\![n]$, 
\begin{align*}
\mbox{(I):}&&\E[\prod_{i\in J} X_i] & \quad \leq \quad \prod_{i\in J} p_i \\
\mbox{(E):}&&\E[\prod_{i\in J} (1-X_i)] & \quad \leq \quad  \prod_{i\in J} (1-p_i)
\end{align*}
where $p_i$ is the probability that key $i$ is included
in the sample $S$ and $\E[\prod_{i\in J} X_i]$
is the probability that all $i\in J$
are included in the sample. Symmetrically
$\E[\prod_{i\in J} (1-X_i)]$ is the probability 
that all $i\in J$ are excluded from the sample. 
 \end{trivlist}

\noindent
{\bf Tail bounds.}
  For both Poisson~\cite{Chernoff} and \varopt~\cite{PancSri:sicomp97,Sri01,GandhiKPS:jacm06,varopt_full:CDKLT10} samples we have
strong tail bounds on the number of samples $J\cap S$ from a subset $J$.
We state the basic Chernoff bounds on this quantity
(other more familiar bounds are derivable from them).
Let $X_i$ be an indicator variable for $i$ being in the sample, so
$X_i = 1$ if $i\in S$ and $0$ otherwise.
Then $X_J$, the number of samples intersecting $J$ is
$X_J=\sum_{i\in J} X_i$, with mean 
  $\mu = \E[X_J]=\sum_{i\in J} p_i$. 

If $\mu \leq a\leq s$, 
the probability of getting more than $a$ samples out of $s$ 
in the subset $J$ is
\begin{equation} \label{chernoff:upper}
\Pr[X_J \ge a] \leq
\left(\frac{s-\mu}{s-a}\right)^{m-a}\left(\frac{\mu}{a}\right)^a
\quad 
\left[\leq
e^{a-\mu}\left(\frac{\mu}{a}\right)^a
\right]\ .
\end{equation}
For $0 \leq a\leq\mu$, the probability of fewer than $a$
samples in $J$ is 
\begin{equation} \label{chernoff:lower}
\Pr[X_J \le a] \leq
\left(\frac{s-\mu}{s-a}\right)^{m-a}\left(\frac{\mu}{a}\right)^a
\quad 
\left[\leq
e^{a-\mu}\left(\frac{\mu}{a}\right)^a
\right]\ .
\end{equation}
\noindent
When sampling with IPPS, 
these bounds on the number of samples also imply tight bounds
on the estimated weight $a(J)$: 
It suffices to consider $J$ such that $\forall i\in J, p_i<1$ 
(as we have the exact weight of keys
with $p_i=1$).  
Then
the \HT\ estimate is
$a(J)= \tau X_J= \tau |J\cap S|$ and thus 
the estimate is guaranteed to be accurate:
\begin{equation} \label{chernoffw:bounds}
\Pr[a(J) \le h],\ \Pr[a(J) \ge h] \leq
e^{(h-w(J))/\tau}(w(J)/h)^{h/\tau}
\ .
\end{equation}

\noindent
{\bf Range discrepancy.}
The {\em discrepancy} of a sample measures the difference between the
number of keys sampled in a range to the number expected to be there. 
Formally, consider keys with attached IPPS inclusion probabilities
$p_i$ over a structure with a set of ranges $\mathcal{R}$. 
The discrepancy $\Delta$ of a set of keys $S$ with respect to range $R$ is

\smallskip
\centerline{$
\Delta(S,R)=\big| |S\cap R| - \sum_{i\in R} p_i \big| .$}

\smallskip
The {\em maximum range discrepancy} of $S$ is accordingly the maximum
discrepancy over $R\in \cR$.
For a sample distribution $\Omega$, we define the
maximum range discrepancy as

\smallskip
\centerline{$
\Delta = \max_{S\in\Omega} \max_{ R\in \cR} \big\vert |S\cap R| - \sum_{i\in R} p_i \big\vert\ $}

\smallskip
The value of the discrepancy has implications for the accuracy of
query answering. 
The absolute error of the \HT\ estimator (\ref{ipps_est_subset}) on $R$ is
the product of $\tau$ and the discrepancy: $\tau \cdot \Delta(S,R)$.
We therefore seek sampling schemes which have a small discrepancy. 

For a Poisson IPPS or \varopt\ sample, the discrepancy on a range $R$ 
is subject to tail bounds (Eqns.~\eqref{chernoff:upper}
and~\eqref{chernoff:lower}) with $\mu=p(R)$ and has expectation $O(\sqrt{p(R)})$.

\medskip
\noindent
{\bf $\epsilon$-approximations.}
Maximum range discrepancy 
is closely related to the concept of an 
{\em $\epsilon$-approximation} \cite{VCepsapprox:1971}.  
A set of points $A$ is an 
$\epsilon$-approximation of the range space $(X,\cR)$, if for
 {\em every} range $R$,

\smallskip
\centerline{
$\left\vert \frac{|A \cap R|}{|A|}-\frac{|X\cap R|}{|X|} \right\vert <
  \epsilon\ .$}

\smallskip
(This is stated for uniform weights but easily generalizes).
An $\epsilon$-approximation of size $s$ has 
maximum range discrepancy $\Delta=\epsilon s$.  
A seminal result by Vapnik and Chervonenkis bounds the maximum
estimation error of a random sample over ranges when the VC dimension
is small: 
\begin{theorem}[from \cite{VCepsapprox:1971}] \label{VCapprox} 
For any range space with VC dimension $d$, there is a constant $c$, 
so that a random sample of size 
$$s = c\epsilon^{-2}(d\log(d/\epsilon) +
\log(1/\delta))$$
\noindent
 is an $\epsilon$-approximation with probability $1-\delta$.  
\end{theorem}
The structures we consider have constant VC dimension, and the theorem
can be proved from a direct application of Chernoff bounds. 
Because \varopt\ samples satisfy Chernoff bounds, they also satisfy
the theorem statement.  
By rearranging and substituting the bound on $\epsilon$, 
we conclude that a Poisson IPPS or a \varopt\ sample of size $s$ has
maximum
 range discrepancy  
$\Delta=O(\sqrt{s\log s})$.  
 with probability ($1-\poly(1/s)$) 

\ignore{ Comment from E:  the structure we consider have a small ``scaffold'' dimension, which is, for a set $S$ of size $s$,  the maximum ratio of the logarithm of the number for distinct ranges induced by $S$ and the logarithm of $s$.  When $d$ is replaced by the scaffold dimension, the theorem follows from Chernoff bounds.  On our structures the VC dimension and the scaffold dimension are roughly the same.}

\ignore{
The smallest discrepancy we can hope for is $\Delta(S,R)<1$, which means
that
\begin{equation} \label{saopt:cond}
\ |S\cap R|=\left\{ \begin{array}{ll} \lfloor \sum_{i\in R} p_i \rfloor, &  \text{prob. } \lceil \sum_{i\in R} p_i \rceil - \sum_{i\in R} p_i \\ \lceil \sum_{i\in R} p_i \rceil, &  \text{prob. } \sum_{i\in R} p_i-\lfloor \sum_{i\in R} p_i \rfloor \end{array}\right. \end{equation}
We shall see that Hierarchy structures admit sample distributions which
satisfy~(\ref{saopt:cond}) for all ranges $R\in\cR$. Generally, however, 
we have to settle for larger values of $\Delta$.

 }

\subsection*{Other Summarization Methods}
\noindent
In addition to sampling methods such as Poisson IPPS and \varopt\, 
there have been many other approaches to summarizing data in range
spaces. 
We provide a quick review of the rich literature on 
summarization methods specifically designed for range-sum queries.
Some methods use random sampling in their construction,
although the resulting summary is not itself a \varopt\ sample. 

\smallskip
\noindent
{\bf $\epsilon$-approximations.}
\ignore{
 Geometric discrepancy theory provides us with bounds on the size
of $\epsilon$-approximations that are generally much superior to what is
obtained via a structure-oblivious sample.
For order or hierarchy structures, 
it is easy to obtain a deterministic 
set with maximum range discrepancy $\Delta=1$, whereas
$\Delta=O(\sqrt{s\log s})$ is obtained via 
a structure-oblivious sample of the same size.

 For higher dimensions, constructions of $\epsilon$-approximations are more complex.
For half spaces on a 
(uniform distribution on the) $d$ dimensional hypercube there is a
lower bound of $\Omega(1/\epsilon^{\frac{2d}{d+1}})$~\cite{Beck87irregularities,Alexander_irregularities:1991}, and a nearly matching upper bound by
Matousek~\cite{Matousek_disc_halfspaces:DCG1995} (cf.~\cite{Matousekbook:geodisc,Chazelle:bookdisc}) followed by algorithmic constructions~\cite{SuriTZ:CG2004,BCEG:ToA2007} (In our terminology, there are distributions where any set of points has  maximum range discrepancy $\Delta=\Omega(s^{\frac{d-1}{2d}})$.)
}
As noted above, $\epsilon$-approximations accurately estimate the
number of points falling in a range. 
For axis-parallel hyper-rectangles, Suri, Toth, and Zhou presented randomized 
constructions that with constant probability 
generates an $\epsilon$-approximation 
of size $O(\epsilon^{\frac{-2d}{d+1}}\log^d(\epsilon^{\frac{2}{d+1}}n))$ and an alternative but computationally intensive construction 
with much better asymptotic dependence on $\epsilon$: $O(\frac{1}{\epsilon}\log^{2d+1} \frac{1}{\epsilon})$~\cite{SuriTZ:CG2004}.   
The best space upper bound we are aware of is
$O(\frac{1}{\epsilon}\log^{2d}(\frac{1}{\epsilon}))$~\cite{Phillips:icalp2008}. 

A proposal in~\cite{BCDHS:KDD2003}
is to construct an $\epsilon$-approximation of a random sample of the dataset
instead of the full dataset.  This is somewhat related to our I/O efficient
 constructions that utilize an initial larger random sample. 
The differences are that we use the sample only as a guide---the final
summary is not necessarily a subset of the initial sample---and that
the result of our construction is a structure-aware \varopt\ sample of
the full data set. 

Another construction \cite{HSST:ISAAC04}
trades better dependence on $\epsilon$ with logarithmic dependence on
domain size.
The data structure is built deterministically by dividing on each
dimension in turn, and retaining ``heavy'' ranges. 
This can be seen as a multi-dimensional variant of the related
q-digest data structure \cite{ShrivastavaBAS:sensys04}.


\ignore{
For $d\geq2$, discrepancy
$\Delta=\Omega(s^{\frac{d-1}{2d}})$ means that the summary 
provides no meaningful approximation for 
ranges of size $o(s^{\frac{d-1}{2d}})$.  
In contrast, Poisson samples have expected range discrepancy 
$O(p(R)^{\frac{1}{2}})$, so for range queries with
$p(R) \ll s^{\frac{1}{2}}$ we get better estimates from these samples.
What we aim for are \varopt\ summaries 
with expected range discrepancy of 
$\min\{p(R)^{\frac{1}{2}},s^{\frac{1}{4}}\}$.
}

\smallskip
\noindent
{\bf Sketches and Projections.}
Random projections are a key tool in dimensionality reduction, which
allows large data sets to be compactly summarized and queried. 
Sketches are particular kinds of random projections, which can be
computed in small space \cite{ccf:icalp2002}.
By keeping multiple sketches of the data at multiple levels of
granularity, we can provide $\epsilon$-approximation-like bounds in
space that depends linearly on $\epsilon^{-1}$ and logarithmically on
the domain size.   

\smallskip
\noindent
{\bf Wavelet transforms and deterministic histograms.}
A natural approach to summarizing large data for range queries is to
decompose the range space into ``heavy'' rectangles. 
The answer to any range query is the sum of weights of all rectangles
fully contained in by the query, plus partial contributions from those
partly overlapped by the query. 
The accuracy then depends on the number (and weight) of rectangles
overlapped by the query. 
This intuition underlies various attempts based on building 
 (multi-dimensional)
histograms~\cite{GKTD:sigmod2000,PoosalaIoannidis:VLDB97,LKC:sigmod1999}.
These tend to be somewhat heuristic in nature, and offer little by way
of guaranteed accuracy.

More generally, we can represent the data in terms of objects other
than rectangles: this yields transformations such as 
DCT, Fourier transforms and wavelet representations. 
Of these, wavelet representations are most popular for representing
range data~\cite{Matias98wavelet-basedhistograms,Vitter98datacube}.
Given data of domain size $u$, the transform generates $u$
coefficients, which are then {\em thresholded}: we pick the $s$
largest coefficients (after normalization) to represent the data. 
When the data is dense, we can compute the transform in time $O(u)$,
but when the domain is large and the data sparse, it is more efficient
to generate the transform of each key, in time proportional to the
product of the logarithms of the size of each dimension per key. 

\ignore{
\medskip
\noindent
{\bf Structure-aware sampling of streams} 
We are currently exploring structure-aware sampling in data streams.
In a stream context, there is a unique \varopt\ sample distribution,
and therefore it is not possible to obtain samples that are both
\varopt\ and structure-aware.  We are exploring tradeoffs between
\varopt\ properties and structure awareness. 
}

\section{Sequences of aggregations}
\label{probaggapp}
\noindent
Our algorithms repeatedly apply probabilistic aggregation:

\begin{lemma}\label{lemma:a}
 Consider a sequence $p^{(0)},p^{(1)},p^{(2)},p^{(3)},\ldots$ where
$p^{(h)}$ is a probabilistic aggregate of $p^{(h-1)}$.
\begin{itemize}
\item
$p^{(h)}_i\in \{0,1\}$ implies $p^{(h+1)}_i\equiv p^{(h)}_i$.  
Thus in a sequence of aggregations, 
any entry that is set remains set, so
the number of positive entries in the output is
at most that in the input.
\item
Probabilistic aggregation is {\em transitive}, that is, if $h_1<h_2$ then
$p^{(h_2)}$ is a probabilistic aggregate of $p^{(h_1)}$.
\end{itemize}
\end{lemma}

\begin{proof}
We show that (I) holds under transitivity.  The proof for (E) is
similar, and the other properties are immediate.
We show that if $p^{(h+2)}$ is an aggregate of $p^{(h+1)}$, and 
$p^{(h+1)}$ is an aggregate of $p^{(h)}$, then
$p^{(h+2)}$ is an aggregate of $p^{(h)}$.
%

$\begin{array}{rcl}
\displaystyle
\E_{p^{(2)}|p^{(0)}}[\prod_{i\in J} p^{(2)}_i] & = & 
\displaystyle
\E_{p^{(1)}|p^{(0)}}\left[ \E_{p^{(2)}|p^{(1)}} [\prod_{i\in J} p^{(2)}_i]\right] \\
& \leq & 
\displaystyle
\E_{p^{(1)}|p^{(0)}} [\prod_{i\in J} p^{(1)}_i] \leq \prod_{i\in J} p^{(0)}_i
\end{array}$
\end{proof}

\section{Multiple ranges in a hierarchy} 
\label{mrangeH:sec}
\notinproc{
\medskip
\noindent
{\bf Multi-range queries.}
We now bound the discrepancy for a query that is a union of
several ranges: 
}
\onlyinproc{
We show that
the discrepancy of a query that spans multiple ranges in a hierarchy
is bounded by the number of ranges. }

 \begin{lemma} \label{hierarcyunionbound:lemma}
Let $Q$ be a union of $\ell$ disjoint ranges $R_1,\ldots,R_\ell$ on a
hierarchy structure. 
The discrepancy is at most $\ell$ and is distributed
as the error of a \varopt\ sample on a subset of size
$\mu=\sum_{h=1}^\ell (p(R_h)-\lfloor p(R_h)\rfloor)\leq \ell$.
\end{lemma}

\begin{proof}
Consider a truncated hierarchy where 
the nodes $R_1,\ldots,R_\ell$ are leaves.  Include other nodes to form
a set $L'$ of disjoint nodes which covers all original leaves.  For
each leaf node in the truncated hierarchy $R_h\in L'$ consider a
corresponding ``leftover'' 0/1 random variable with 
probabilities $p(R_h)-\lfloor p(R_h)\rfloor$: 
its value is 1 if the range $R_h$ has $\lceil R_h \rceil$ samples, 
and its value is $0$ if there are
$\lfloor R_h \rfloor$ samples. 
It follows directly from our construction that 
the sample with respect to these leftovers is a \varopt\ sample, since
the original summarization from $L'$ up proceeds like a hierarchy
summarization, treating the leftovers as leaves. 
\end{proof}

Applying Chernoff bounds, we obtain that 
the discrepancy on $Q$ is at most $\sqrt{\ell}$ with high probability.

\section{Order Structures} \label{ordered:sec}
\noindent
Recall that order structures consist of all intervals of keys.
\notinproc{
\begin{theorem} \label{orderoffline:thm}
For the order structure (all intervals of ordered keys), 
there always exists a \varopt\ sample
distribution with maximum range discrepancy $\Delta\leq 2$.  For any
fixed $\Delta<2$, there exist inputs for which a \varopt\ sample
distribution with maximum range discrepancy $\leq \Delta$ does not
exist. 
\end{theorem}

We establish the positive part of the Theorem by presenting 
an efficient summarization algorithm.
Appendix~\ref{negorderoffline} shows the negative part 
 by demonstrating a family of hard input instances.

 \medskip
}\onlyinproc{
\begin{theorem}[Theorem \ref{orderoffline:thm} restated]
For the order structure
(i) there always exists a \varopt\ sample
distribution with maximum range discrepancy $\Delta\leq 2$.  
(ii) For any
fixed $\Delta<2$, there exist inputs for which a \varopt\ sample
distribution with maximum range discrepancy $\leq \Delta$ does not
exist. 
\end{theorem}
}

\noindent
{\bf Order Structure Summarization.}
 To gain intuition, first consider inputs where
there is a partition $\cL$ of the ordered keys into 
non-overlapping intervals such that for each interval
$J\in\cL$,  $p(J)=\sum_{i\in J} p_i \equiv 1$, i.e. the initial
probabilities sum to 1.
In this case, there are \varopt\ samples which pick one key from each 
interval $J \in \cL$. 
Now observe that any query interval $R$ covers some number
of full unit intervals.
The only discrepancy comes from the at most 2 end intervals, and so
the maximum range discrepancy is bounded by the probability mass
therein,  $\Delta < 2$. 
Therefore, this sample has maximum range discrepancy $\Delta< 2$.  

The general case is handled by
{\sc OSsummarize}$(p_1,\ldots,p_n)$.
(Algorithm~\ref{orderoffline:alg}). 
Without loss of generality, key $i$ is the $i$th key in the sorted
order, and $p_i$ is its inclusion probability.
The algorithm processes the keys in sorted order, maintaining 
an active key $a$ that is the only key that is not set 
from the prefix processed so far.
 At any step, if there is an active key $a$, the selected pair for the
 aggregation consists of $a$ and the current key.  
Otherwise, the aggregation involves the current and the next key.
The final sample $S$ is the set of keys with  $p_i=1$.  
We now argue it has bounded discrepancy.

\begin{proof}[of Theorem \ref{orderoffline:thm} (i)]
The algorithm can be viewed as a special case of a hierarchy summarization
algorithm where the ordered keys are arranged as a  
hierarchy which is a long path with a single leaf hanging out of each
path node.   
The internal nodes in the hierarchy correspond to prefixes of the
sorted order, and thus they are estimated optimally
(the number of samples is the floor or ceiling of the expectation):
For any $i$, 
the number of members of $S$ amongst the first $i$ keys is
$$|S\cap[1,i]|\in \{\lfloor \sum_{j\leq i} p_i \rfloor,\lceil
\sum_{j\leq i} p_i \rceil\}\ .$$ 
For a key range $R=[i_1,i_2]$ that is not a prefix ($i_2\geq i_1>1$),
we can express it as the difference of two prefixes:
{\small
\begin{eqnarray*}
\lefteqn{|S\cap R| = |S\cap [1,i_2]| - |S\cap [1,i_1-1]|}\\
& \leq & \lceil \sum_{j\leq i_2} p_i \rceil - \lfloor \sum_{j\leq i_1-1} p_i \rfloor \leq 1+\sum_{j\leq i_2} p_i - (-1+ \sum_{j\leq i_1-1} p_i) \\
& = & 2+\sum_{i_1\leq j \leq i_2}p_i\ .\\
 & \geq & \lfloor \sum_{j\leq i_2} p_i \rfloor - \lceil \sum_{j\leq i_1-1} p_i \rceil \geq -1+\sum_{j\leq i_2} p_i - (1+ \sum_{j\leq i_1-1} p_i) \\
 & =  & -2+\sum_{i_1\leq j \leq i_2}p_i\ .
\end{eqnarray*}}
\noindent
Hence the maximum discrepancy is at most 2, as claimed. 
\end{proof}

\noindent
{\bf Summaries with Smaller discrepancy.}
  Requiring the summary to be \varopt\ means that it may not be
feasible to guarantee $\Delta$ {\em strictly} less than 2. 
We can, however, obtain a {\em deterministic} set with maximum range 
discrepancy $\Delta < 1$:
 Associate key $i$ with the
interval $H_i=(\sum_{j< i} p_j, \sum_{j\leq i} p_j]$ on the positive axis (with respect to the original vector of inclusion probabilities) and simply 
include in $S$ all keys where the $H_i$ interval contains an integer.
In fact, we can obtain a sample which satisfies the
\varopt\ requirements (i) and (ii) but not (iii) with $\Delta < 1$:
Uniformly pick $\alpha\in [0,1]$
and store all keys $i$ so that $h+\alpha \in H_i$ for each integer $h$.
 This sampling method is known as {\em systematic sampling}~\cite{SSW92}. 
Systematic samples, however, suffer from positive correlations which mean
that estimates on some subsets have high variance (and Chernoff tail
bounds do not apply). 

\label{negorderoffline}

\begin{algorithm}[t]
\caption{{\sc OSsummarize}$(p_1,\ldots,p_n)$}\label{orderoffline:alg}
\begin{algorithmic}[1]
\State $a \leftarrow 1$ \Comment{leftover key}
\State $i = 2$ \Comment{current  key}

\While {$i \leq n$} \Comment{left to right scan of keys}
  \While {$p_a = 1$ and $a< n$}
    \State  $a$++;
  \EndWhile

  \Statex
  \State $i\leftarrow a+1$
  \While {$p_i =1$ and $i< n$}
    \State $i$++;
  \EndWhile

  \Statex

  \If {$p_a < 1$ and $p_i < 1$}
    \State {\sc Pair-Aggregate}$(a,i)$
    \If {$p_a = 1$ or $p_a = 0$} \Comment{$p_a$ is set}
      \State $a \leftarrow i$ \Comment{$i$ is the new leftover key}
    \EndIf
    \State $i$++
  \EndIf 
\EndWhile
\end{algorithmic}
\end{algorithm}

\medskip
\noindent
{\bf Lower bound on discrepancy.}
We show that there are ordered weighted sets 
for which we can not obtain a \varopt\ summary with 
maximum range discrepancy $\Delta < 2$.

\begin{proof}[of Theorem~\ref{orderoffline:thm} (ii)]

For any positive integer $m$, 
we show that for some sequence, there is no \varopt\ sample with
$\Delta \leq 2-1/m$.

We use a sequence where $p_i= \epsilon \ll 1/(4m)$ and $\sum_i p_i \geq 5m$.
Let
$i_1<i_2<i_3,\cdots$ be the included keys, sorted by order.
With each key $i_\ell$ we associate the position 
$s(i_\ell)=\sum_{j\leq i_\ell} p_j$.

 We give a proof by contradiction.  
Consider keys $i_\ell,i_j$ with $\ell<j$.
If an interval contains at most $h$ sampled keys, it must be of size
at most $h+\Delta$.
If an interval contains at least $h$ sampled keys,  it must be of
size at least $h-\Delta$.

The interval $[s(i_\ell)-\epsilon,s(i_j)]$, which is of size
$s(i_j)-s(i_\ell)+\epsilon$, 
contains $j-\ell+1$ sampled keys.

We obtain that $s(i_j)-s(i_\ell)+\epsilon \geq j-\ell+1-\Delta$.
Rearranging, $s(i_j) \geq s(i_\ell)-\epsilon-\Delta+1+j-\ell$.

The interval $(s(i_\ell),s(i_j)-\epsilon)$, which is of size
$s(i_j)-s(i_\ell)-\epsilon$, contains 
$j-\ell-1$ sampled keys.
Hence,
$s(i_j)-s(i_\ell)-\epsilon \leq j-\ell-1+\Delta$.
Rearranging,
$s(i_j) \leq s(i_\ell)+\epsilon+j-\ell+\Delta-1$.

 From the above, 
for  $j>\ell$, we have
\begin{equation} \label{contrell}
 s(i_j) \in (s(i_\ell)+j-\ell-\Delta+1-\epsilon,s(i_\ell)+j-\ell+\Delta+1+\epsilon)\ .
\end{equation}

For $j\geq 2$, 
\begin{equation} \label{contr1}
s(i_j) \in (s(i_1)-\epsilon+j-\Delta,s(i_1)+j-2+\Delta+\epsilon)\ .
\end{equation}

 Fixing the first $j$ included keys, $i_1,\ldots,i_j$, 
the conditional probability on $i_{j+1}=h$ is at most $p_h$.
We have
$s(i_{j+1})\in (s(i_j)+2-\Delta-\epsilon, s(i_j)+\Delta+\epsilon)$.
Therefore, there must be a positive probability for the event
$$s(i_{j+1})< s(i_j)+\Delta+\epsilon-(1-\epsilon)= s(i_j)+1-1/m+2\epsilon$$
which is contained in the event
$s(i_{j+1})  \leq s(i_j)+1-1/(2m)$.

 Iterating this argument, we obtain that for all $k>1$,
we must have positive probability
for $s(i_k) < s(i_1) + (k-1)(1-1/(2m))=s(i_1) +k-1 -(k-1)/2m$.  
From (\ref{contr1}) 
we have $s(i_k)\geq s(i_1) + k-1 -(1-1/m)=i_1+k-2+1/m$.
Taking $k=4m$, we get a contradiction.
\end{proof}

\section{Analysis of KD-hierarchy} \label{KDanal:sec}
\eat{
\begin{figure}[t]
\centering
\includegraphics[width=0.8\columnwidth]{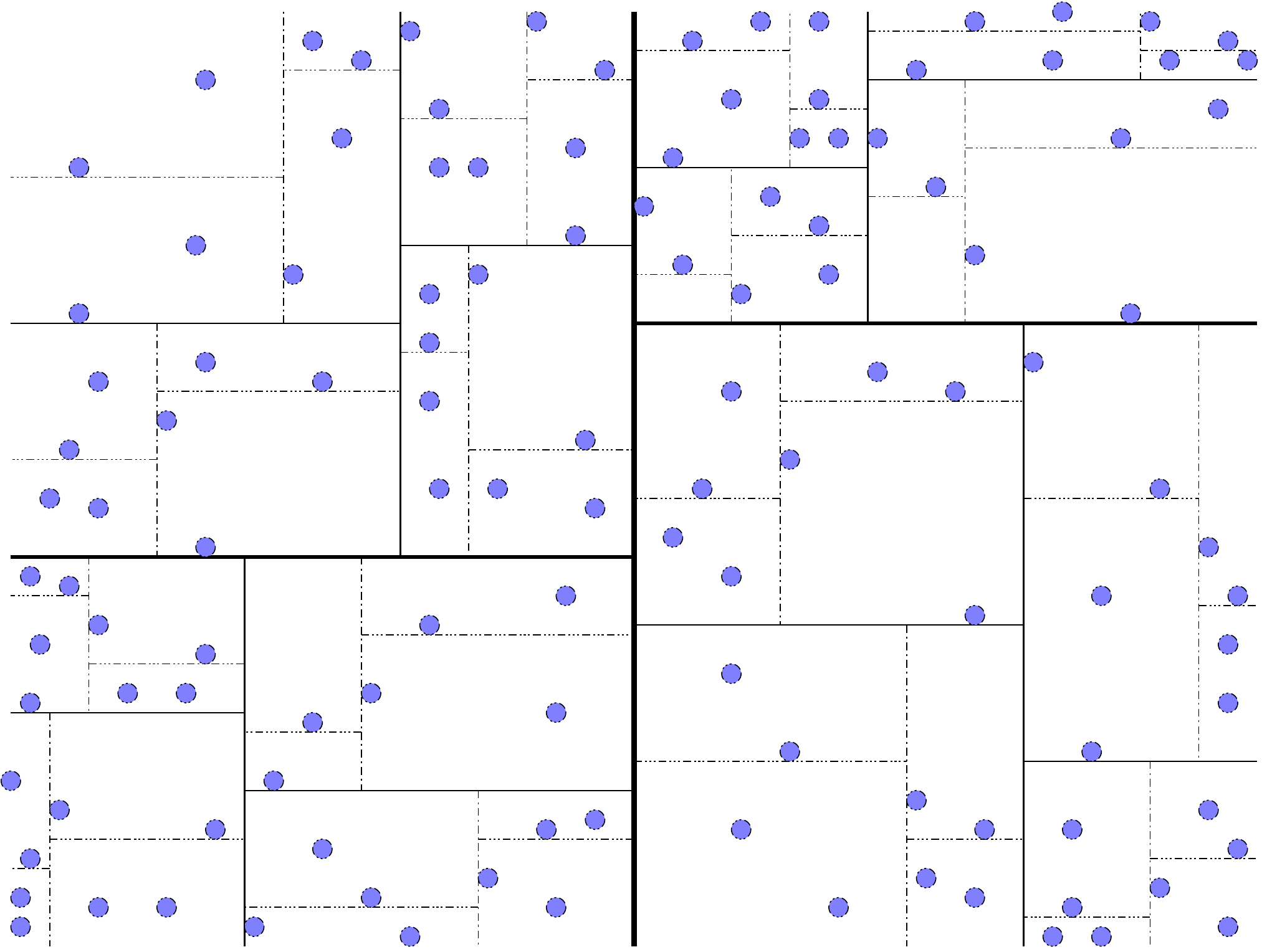} \\
\includegraphics[width=0.8\columnwidth]{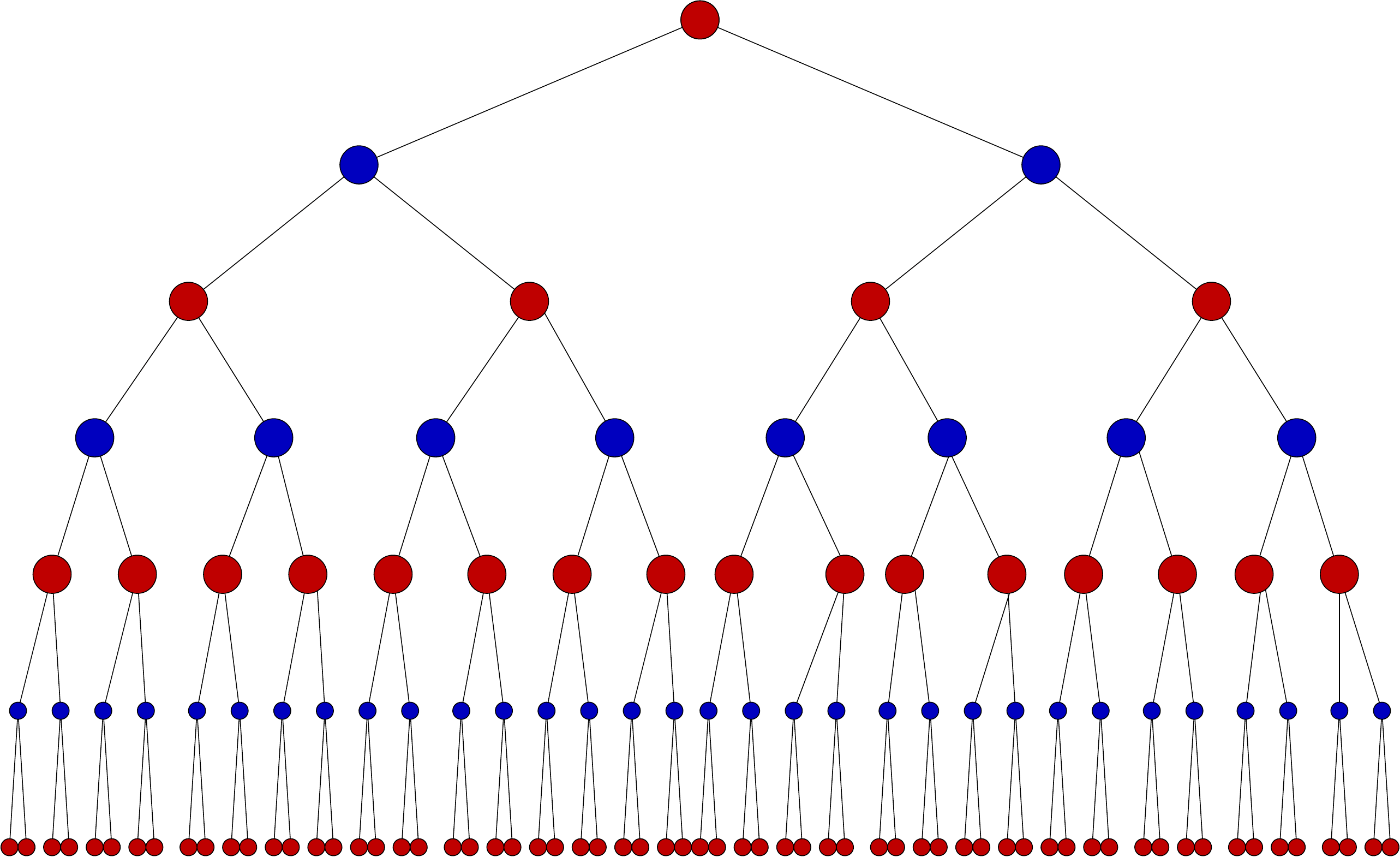}
\caption{KD Hierarchical partition of two dimensional data}
\label{fig:kdpart}
\end{figure}

\begin{figure}[t]
\centering
\includegraphics[width=0.8\columnwidth]{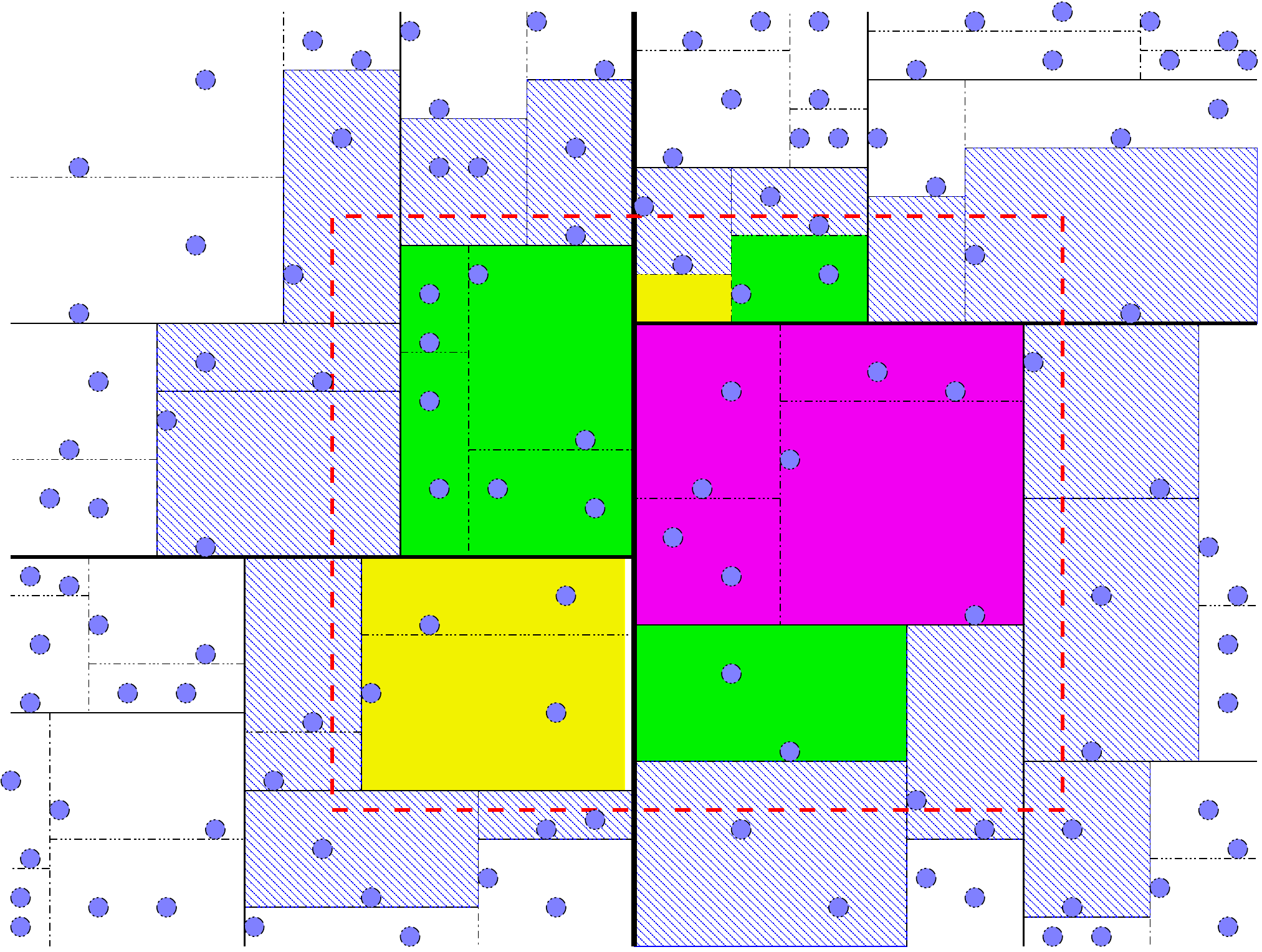} \\
\includegraphics[width=0.8\columnwidth]{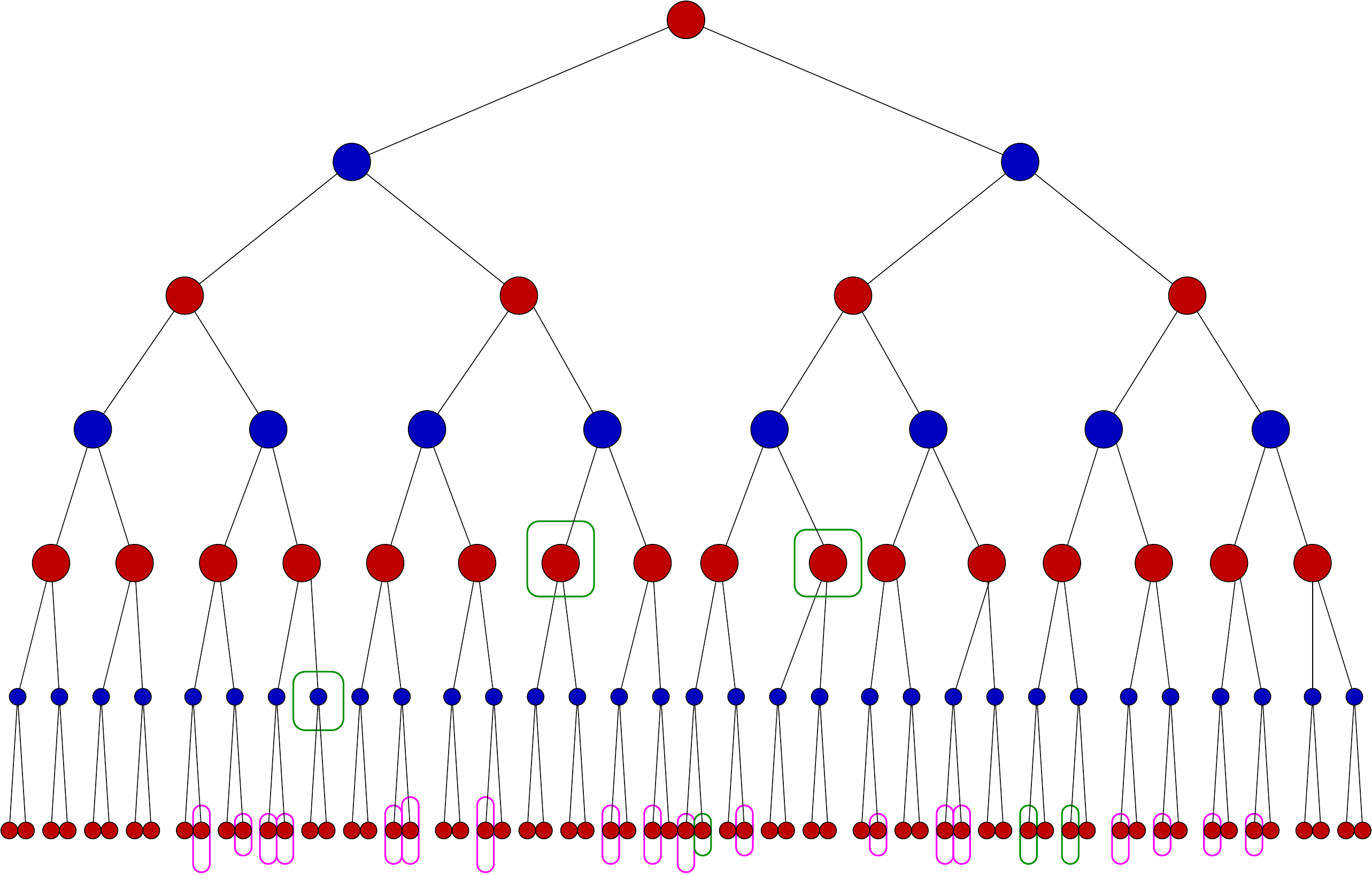}
\caption{Query rectangle on a hierarchical partition}
\label{fig:kdpart_query}
\end{figure}
}
\begin{figure}[t]
\centering
\begin{tabular}{ccc}
\includegraphics[width=0.45\columnwidth]{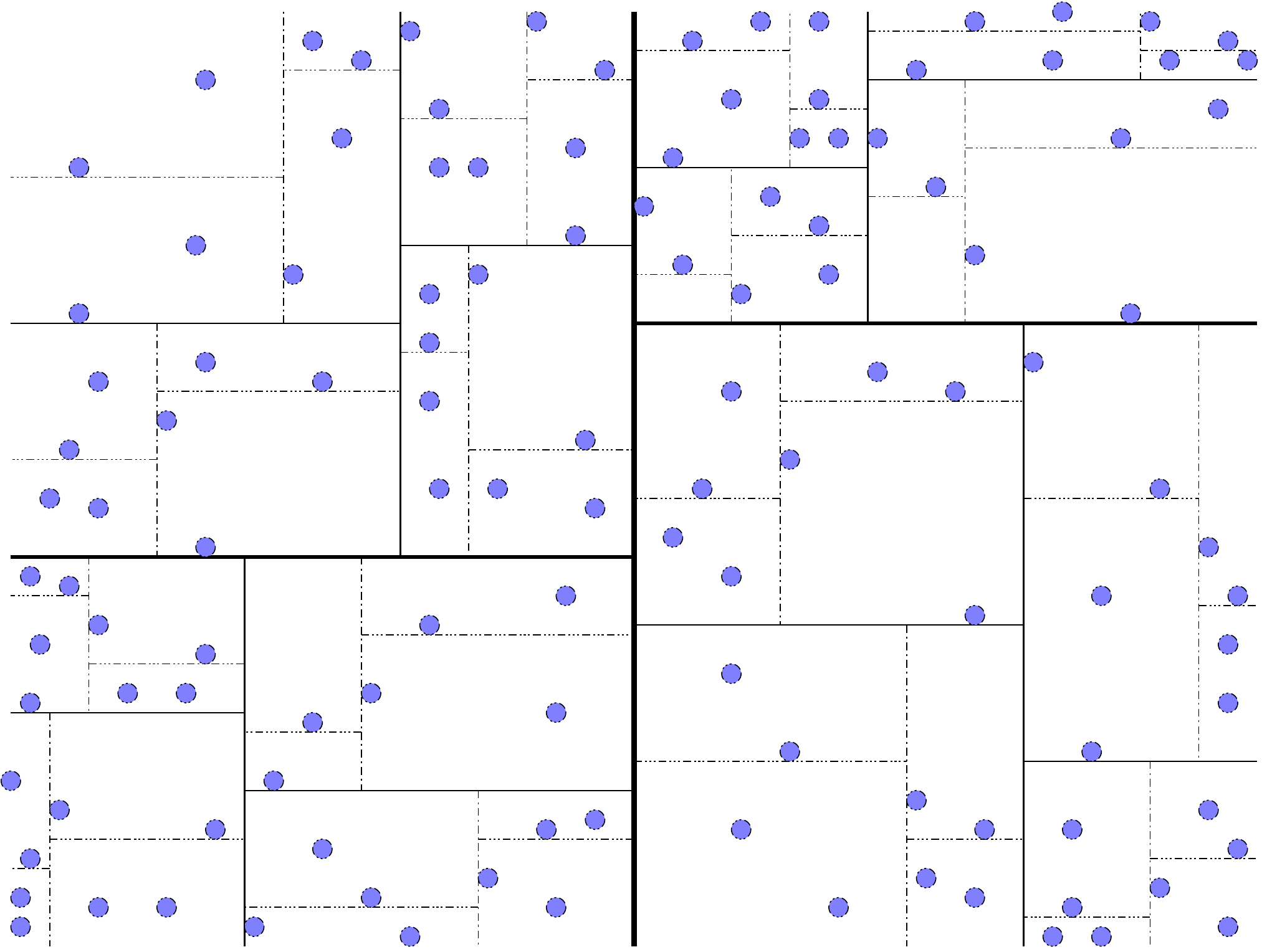} & $\quad\, $ &
\includegraphics[width=0.45\columnwidth]{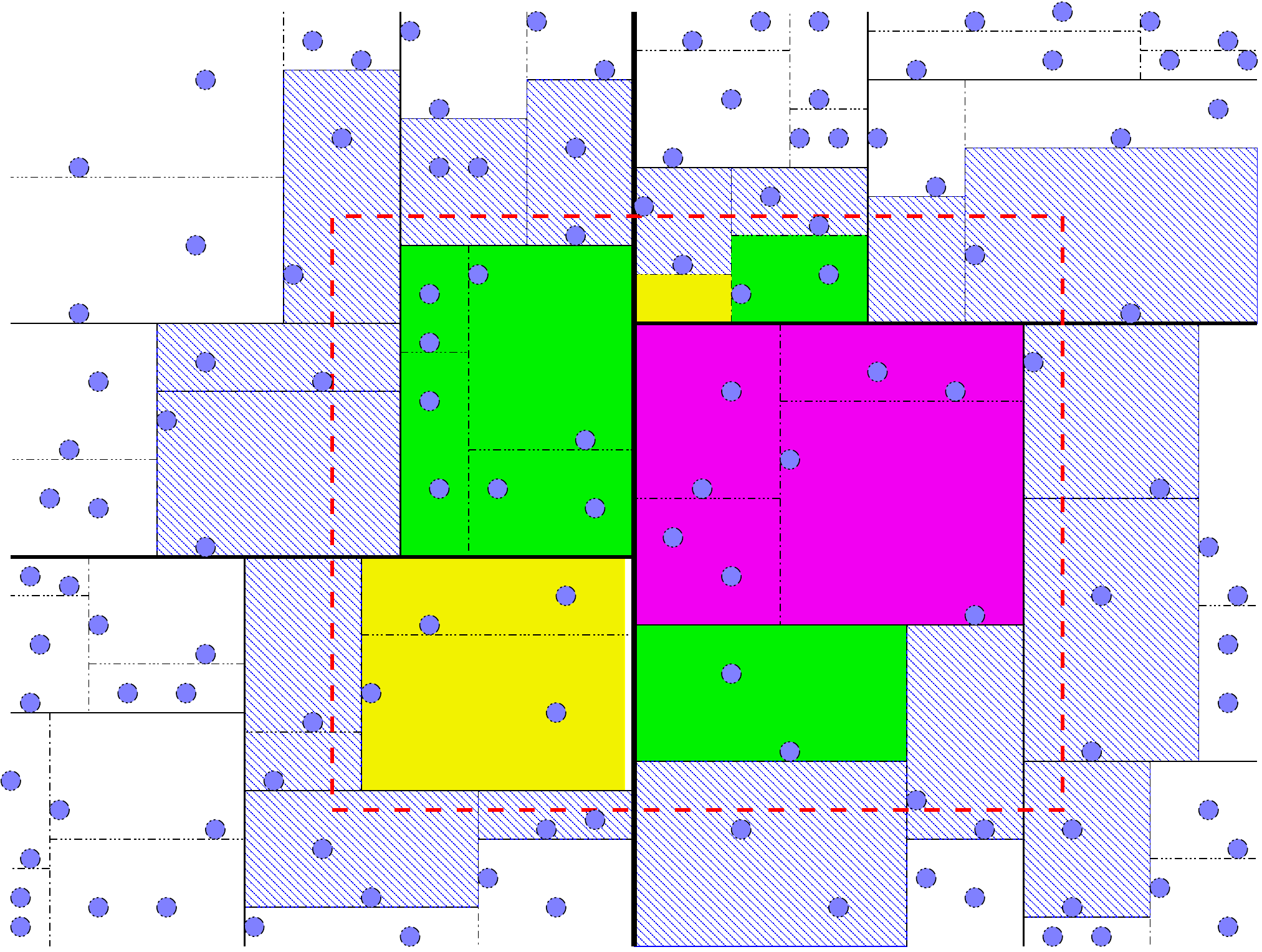}\\
\subfigure[Partition]{\includegraphics[width=0.45\columnwidth]{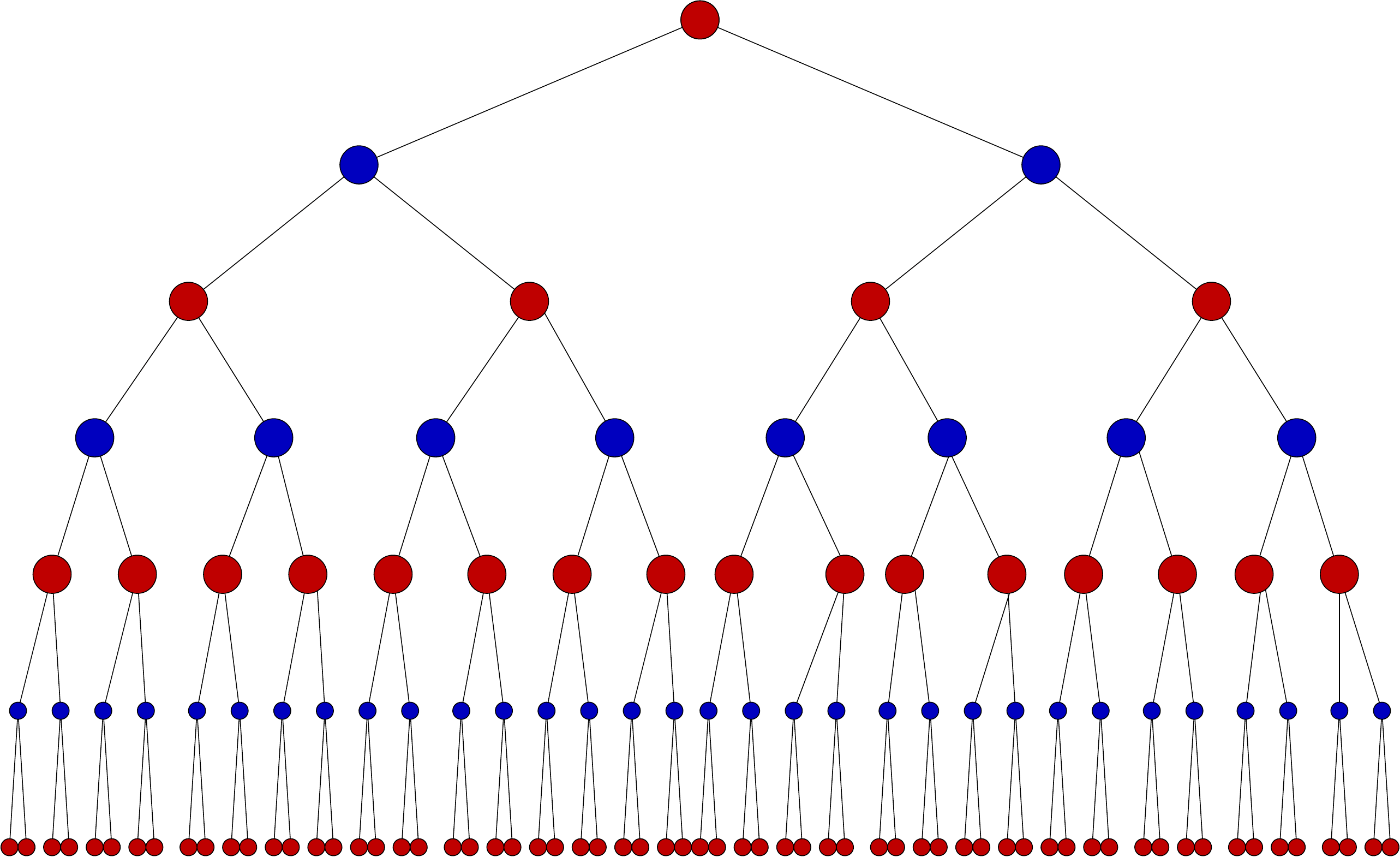}\label{fig:kdpart}} & $\quad\, $ &
\subfigure[Query rectangle]{\includegraphics[width=0.45\columnwidth]{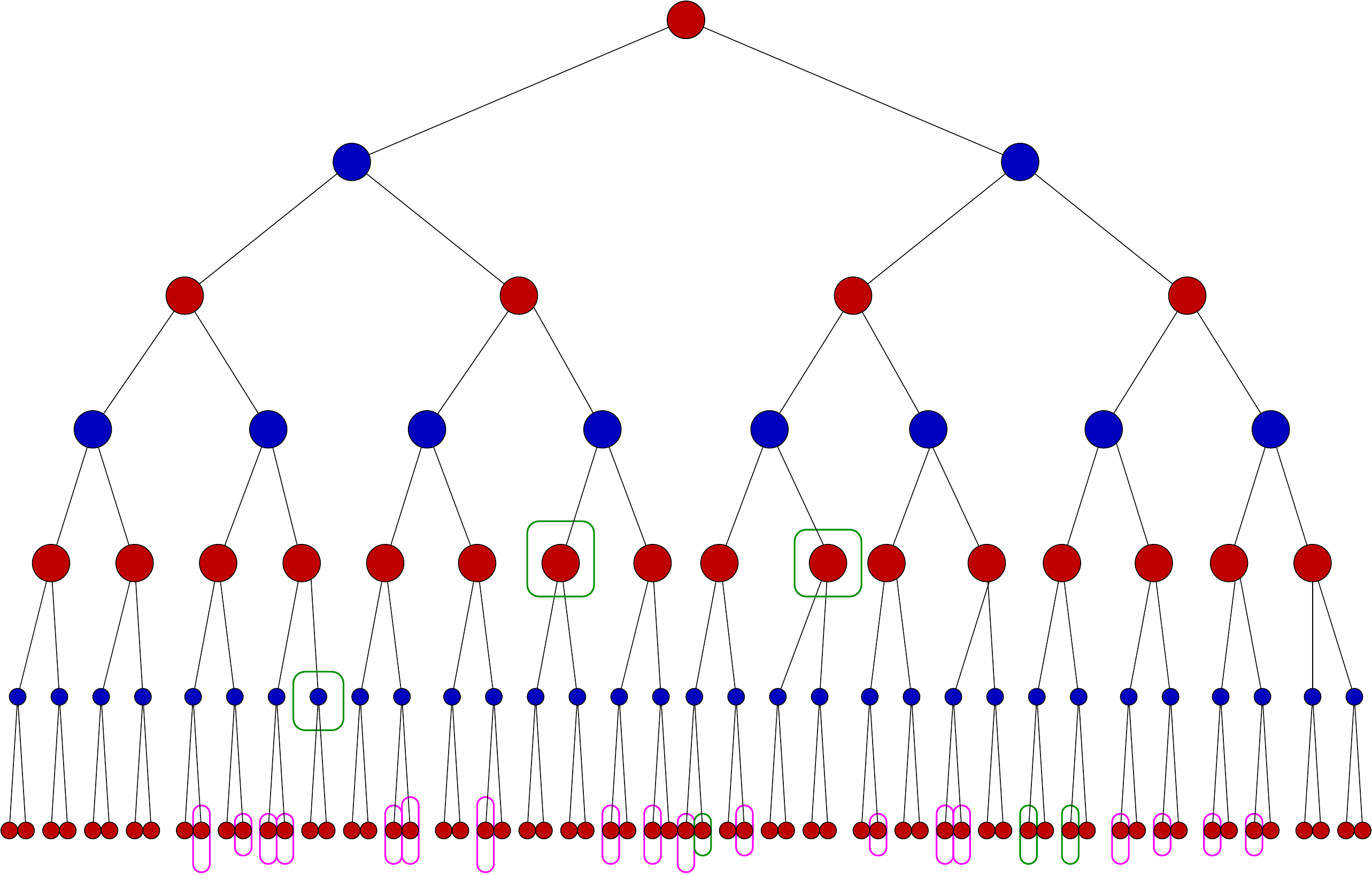}\label{fig:kdpart_query}}
\end{tabular}
\caption{KD Hierarchical partition of two dimensional data}
\label{fig:kdpart_all}
\end{figure}

The {\sc KD-hierarchy} algorithm (Algorithm \ref{kdhierarchy:alg})
aims to minimize the discrepancy within a product space. 
Figure~\ref{fig:kdpart_all}\ref{fig:kdpart} shows a two-dimensional set of $64$ keys that
are uniformly weighted, with sampling probabilities $1/2$, and the
corresponding kd-tree: a balanced binary tree of depth 6.  
The cuts alternate vertically (red tree nodes) and horizontally
(blue nodes).  Right-hand children in tree correspond to right/upper
parts of cuts and left-hand children to left/lower parts.

We now analyze the resulting summary, based on the properties of the
space partitioning performed by the kd-tree.
We use $v$ to refer interchangeably to a node in the tree and the  
hyperrectangle induced by node $v$ in the tree.  
A node $v$ at depth $d$ in the tree has probability mass
  $p(v)\leq s/2^d+2$.
 We refer to the set of minimum depth nodes that satisfy $p(v)\leq 1$
 as s-leaves (for {\em super leaves})
($v$ is an s-leaf iff $p(v)\leq 1$ and its immediate ancestor
 $a(v)$ has $p(a(v))>1$.  
The depth of an s-leaf (and of the hierarchy when truncated at
s-leaves) is at most $D=2+\lceil \log_2 s \rceil$.
\begin{lemma}
Any axis-parallel 
hyperplane cuts $O(s^{\frac{d-1}{d}})$ s-leaves.  
\label{lem:bsize}
\end{lemma}
\begin{proof}
Consider the hierarchy level-by-level, top to bottom.  
Each level where the axis was not perpendicular to the hyperplane at most
 doubles the number of nodes that intersect the hyperplane. 
Levels where the partition axis is perpendicular to the hyperplane do not
increase the number of intersecting nodes.
Because axes were used in a round-robin fashion, 
the fraction of levels that can double the number of intersecting nodes 
is $\frac{d-1}{d}$.
Hence, when we reach the s-leaf level, 
the number of intersecting nodes is at most  
$2^{\frac{d-1}{d} D} = O(s^{\frac{d-1}{d}})$.
\end{proof}
An immediate corollary is that
the boundary of an axis-parallel box $R$ may intersect at most $2d
s^{\frac{d-1}{d}}$ s-leaves. We denote by $B(R)$ this set of
boundary s-leaves. 
 Let $U(R)$ be a minimum size collection of nodes in the hierarchy 
such that no internal node contains a leaf from $B(R)$. 
Informally, $U(R)$ consists of the (maximal) hyperrectangles which
are fully contained in $R$ or fully disjoint from $R$.
Figure~\ref{fig:kdpart_all}\ref{fig:kdpart_query} illustrates a query rectangle $R$ (dotted red line) over the data set.
The maximal interior nodes contained in $R$  
($v\in U(R) | v\subset R$) are marked in solid
colors (and green circles in the tree layout) and the boundary s-leaves $B(R)$ in light stripes (magenta circles in the tree layouts).
For example, the magenta rectangle corresponds to the R-L-L-R path.
\begin{lemma}
\label{lem:usize}
The size of $U(R)$ is at most
 $O(d s^{\frac{d-1}{d}}\log s)$.
\end{lemma}
\begin{proof}
  Each node in $U$ must have a 
sibling such that the sibling, or some of its descendants, are in $B(R)$.
If this is not the case, then the two siblings can be replaced by their parent, decreasing the size of $U(R)$, which contradicts its minimality.
We bound the size of $U(R)$ by bounding the number of potential siblings.
The number of ancestors of each boundary leaf is at most the depth 
which is $\leq 2+\lceil \log_2 s\rceil$.  Thus, 
the number of potential siblings is at most the number of boundary leaves times the depth.  By substituting a bound on $|B(R)|$, we obtain the stated upper bound.
\end{proof}
\noindent
These lemmas allow us to bound the estimation error, by applying 
Lemma~\ref{hierarcyunionbound:lemma}.
That is, 
for each $v\in U(R)$ such that $v\subset R$ we have a 0/1 random
variable that is 1 with probability $p_v-\lfloor p_v \rfloor$  and is
$0$ otherwise 
(The value is $0$ if $v$ includes  
$\lfloor p_v \rfloor$ samples and $1$ otherwise).
 For each $v\in B(R)$, we have a random variable that is 1 with
 probability $p(v\cap R)$.  
This is the probability that $S$ contains one key from $v\cap R$ ($S$ can contain at most one key from each s-leaf). 
The sample is \varopt\ over these  random variables with

\vspace*{-1ex}
{\small
$$\mu = \sum_{v\in U(R) | v\subset R} (p(v)-\lfloor p(v) \rfloor) + \sum_{v\in B(R)} p(v\cap R) \leq |U(R)|+|B(R)|\ .$$}

\noindent
Substituting our bounds on $|U(R)|$ and $|B(R)|$ from Lemmas
\ref{lem:bsize} and \ref{lem:usize} gives  accuracy bound
claimed at the start of the section. 


\end{document}